\documentclass[11pt]{article}

\usepackage{graphicx,epsfig,psfrag}
\usepackage[usenames]{color}
\usepackage{subfigure}
\usepackage[square,numbers]{natbib}
\graphicspath{{figs/}}

\usepackage[today,nofancy]{svninfo}
\svnInfo $Id: wsel.tex 119 2006-06-22 21:18:31Z spav $

\usepackage{fancyhdr}
\usepackage{lastpage}
\usepackage{alg}
\floatstyle{ruled}
\restylefloat{algorithmfloat}
\usepackage[newitem,newenum,increaseonly]{paralist}

\usepackage[environments,commands]{comcom}

\providecommand{\Xsym}{X}
\providecommand{\ysym}{y}

\providecommand{\Wsym}{W}
\providecommand{\Usym}{U}
\providecommand{\Qsym}{Q}
\providecommand{\Vsym}{V}
\providecommand{\Tsym}{T}
\providecommand{\Psym}{P}

\providecommand{\tsym}{t}
\providecommand{\qsym}{q}
\providecommand{\rsym}{r}

\providecommand{\wsym}{w}
\providecommand{\vsym}{v}

\providecommand{\jx}{j}

\providecommand{\lasym}{\lambda}
\providecommand{\Lasym}{\Lambda}
\providecommand{\gasym}{\gamma}
\providecommand{\Gasym}{\Gamma}

\providecommand{\scsym}{c}

\providecommand{\Xesym}{E}
\providecommand{\yesym}{f}
\providecommand{\resym}{\epsilon}

\providecommand{\rosym}{\rho}
\providecommand{\alsym}{\alpha}

\providecommand{\regasym}{\alpha}
\providecommand{\regbsym}{\beta}

\providecommand{\ffsym}{\phi}
\providecommand{\obfsym}{\psi}
\providecommand{\mrsym}{\kappa}

\providecommand{\ro}{m}
\providecommand{\cl}{n}
\providecommand{\fc}{l}
\providecommand{\ms}{k}

\providecommand{\jacb}[2]{\prbypr{#1}{#2}}
\providecommand{\kron}[1]{\MATHIT{\neSUB{\delta}{#1}}}

\providecommand{\scalk}[2][{k}]{\MATHIT{\neSUP{#2}{\wrapNeParens{#1}}}}
\providecommand{\Mtxk}[2][{k}]{\scalk[#1]{\Mtx{#2}}}
\providecommand{\vectk}[2][{k}]{\scalk[#1]{\vect{#2}}}

\providecommand{\undrl}[1]{\tilde{#1}}				

\providecommand{\nePAIR}[2]{{#1}\lMr{#2}}

\providecommand{\sclr}[1][{}]{\MATHIT{\neSUB{\scsym}{#1}}}

\providecommand{\LAM}{\Mtx{\Lasym}}
\providecommand{\LAMk}[1]{\MATHIT{\neSUP{\Mtx{\Lasym}}{#1}}}
\providecommand{\lame}[1]{\MATHIT{\neSUB{\lasym}{#1}}}
\providecommand{\lamex}[2]{\MATHIT{\neUL{\lasym}{#1}{#2}}}
\providecommand{\vlam}{\vect{\lasym}}
\providecommand{\dlam}[1]{\jacb{#1}{\vlam}}
\providecommand{\glam}[1][{}]{\MATHIT{\nabla_{\vlam}{#1}}}
\providecommand{\vlamk}[1][{k}]{\vectk[#1]{\lasym}}

\providecommand{\GAM}{\Mtx{\Gasym}}

\providecommand{\GAMj}[1][{j}]{\Mtxk[#1]{\Gasym}}
\providecommand{\game}[1]{\MATHIT{\neSUB{\gasym}{#1}}}

\providecommand{\vgam}{\vect{\gasym}}

\providecommand{\vrok}[1][{k}]{\vectk[#1]{\rosym}}
\providecommand{\salk}[1][{k}]{\scalk[#1]{\alsym}}

\providecommand{\mrp}[2][{p}]{\MATHIT{\neSUB{\mrsym}{#1}\wrapNeParens{#2}}}
\providecommand{\mrpq}[2][{p,q}]{\MATHIT{\neSUB{\mrsym}{#1}\wrapNeParens{#2}}}

\providecommand{\obf}[1]{\MATHIT{\obfsym\wrapNeParens{#1}}}

\providecommand{\eye}[1][{}]{\MATHIT{\neSUB{\Mtx{I}}{#1}}}

\providecommand{\ip}[2]{\MATHIT{\trans{#1}{#2}}}

\providecommand{\prd}{\Mtx{\Xsym}}
\providecommand{\rsp}{\vect{\ysym}}

\providecommand{\Xe}{\Mtx{\Xesym}}
\providecommand{\ye}{\vect{\yesym}}

\providecommand{\msz}[1]{\MATHIT{\neSUB{\ro}{#1}}}

\providecommand{\Xcal}{\MATHIT{\prd_{\text{cal}}}}
\providecommand{\ycal}{\MATHIT{\rsp_{\text{cal}}}}
\providecommand{\Xtst}{\MATHIT{\prd_{\text{tst}}}}
\providecommand{\ytst}{\MATHIT{\rsp_{\text{tst}}}}

\providecommand{\rotst}{\MATHIT{\ro_{\text{t}}}}

\providecommand{\regap}[1]{\vect{\regasym}\wrapNeParens{#1}}
\providecommand{\regbp}[1]{\vect{\regbsym}\wrapNeParens{#1}}

\providecommand{\rega}{\regap{}}
\providecommand{\regb}{\regbp{}}
\providecommand{\rego}{\MATHIT{b_0}}

\providecommand{\dregb}{\dlam{\regb}}
\providecommand{\drega}{\dlam{\rega}}

\providecommand{\regmodel}{\tuple{\regb,\rego}}

\providecommand{\Mij}[2]{\MATHIT{\neSUB{\Mtx{M}}{\nePAIR{#1}{#2}}}}
\providecommand{\Mk}[1][{k}]{\Mtxk[#1]{M}}
\providecommand{\Lk}[1][{k}]{\Mtxk[#1]{L}}
\providecommand{\Xk}[1][{k}]{\Mtxk[#1]{\Xsym}}
\providecommand{\yk}[1][{k}]{\vectk[#1]{\ysym}}

\providecommand{\XLXG}[1][2]{\Xk[0]\LAMk{#1}\trans{\Xk[0]}\GAM}

\providecommand{\vz}{\vect{0}}
\providecommand{\Mz}{\Mtx{0}}

\providecommand{\Wk}[1][{k}]{\vectk[#1]{\Wsym}}
\providecommand{\Vk}[1][{k}]{\vectk[#1]{\Vsym}}
\providecommand{\Tk}[1][{k}]{\vectk[#1]{\Tsym}}
\providecommand{\Pk}[1][{k}]{\vectk[#1]{\Psym}}

\providecommand{\Uk}[1][{k}]{\vectk[#1]{\Usym}}
\providecommand{\Qk}[1][{k}]{\vectk[#1]{\Qsym}}

\providecommand{\Wmt}{\Mtx{\Wsym}}

\providecommand{\Tmt}{\Mtx{\Tsym}}
\providecommand{\Pmt}{\Mtx{\Psym}}

\providecommand{\qk}[1][{k}]{\MATHIT{\neSUB{\qsym}{#1}}}
\providecommand{\tk}[1][{k}]{\MATHIT{\neSUB{\tsym}{#1}}}
\providecommand{\rk}[1][{k}]{\MATHIT{\neSUB{\rsym}{#1}}}

\providecommand{\wk}[1][{k}]{\MATHIT{\neSUB{\wsym}{#1}}}
\providecommand{\vk}[1][{k}]{\MATHIT{\neSUB{\vsym}{#1}}}

\providecommand{\uTk}[1][{k}]{\vectk[#1]{\undrl{\Tsym}}}
\providecommand{\uPk}[1][{k}]{\vectk[#1]{\undrl{\Psym}}}
\providecommand{\uqk}[1][{k}]{\MATHIT{\neSUB{\undrl{\qsym}}{#1}}}

\providecommand{\qvc}{\vect{\qsym}}

\providecommand{\wun}[1]{\MATHIT{\neSUB{\vect{1}}{#1}}}
\providecommand{\Wun}[1]{\vectk[#1]{1}}
\providecommand{\wunt}{\wun{\text{tst}}}

\providecommand{\wunro}{\wun{\ro}}

\providecommand{\resd}{\vect{\resym}}
\providecommand{\rese}[1]{\MATHIT{\neSUB{\resym}{#1}}}
\providecommand{\resdk}[1][{k}]{\vectk[#1]{\resym}}

\providecommand{\diag}[1]{\MATHIT{\text{diag}\wrapNeParens{#1}}}
\providecommand{\ffn}[1]{\MATHIT{\ffsym\wrapNeParens{#1}}}

\providecommand{\MSEP}{\text{MSEP}\xspace}
\providecommand{\RMSEP}{\text{RMSEP}\xspace}
\providecommand{\MSECV}{\text{MSECV}\xspace}
\providecommand{\RMSECV}{\text{RMSECV}\xspace}

\providecommand{\RMSEMCCV}{\text{RMSEMCCV}\xspace}

\providecommand{\SRCEK}{\text{SRCEK}\xspace}

\providecommand{\Qs}[1][{k}]{\MATHIT{Q^2}}

\providecommand{\BIC}{\text{BIC}\xspace}
\providecommand{\aBIC}{\text{aBIC}\xspace}

\providecommand{\sigmy}{\MATHIT{\sigma^2_{\ysym}}}

\providecommand{\dsm}{\emph{Moisture}}
\providecommand{\dsk}{\emph{Kalivas}}
\providecommand{\dsa}{\emph{Artificial}}

\newcommand{\keywords}[1]%
{\def\@keywords{\textbf{#1}}}

\title{\SRCEK{}: A Continuous Embedding of the Channel Selection Problem for WPLS Modeling.}
\author{Steven E. Pav\thanks{spav@alumni.cmu.edu This method of wavelength 
selection may be covered by patent. \cite{USPAT8112375}
This research was originally conducted during the author's tenure at 
Nellcor, a subsidiary of Tyco Healthcare, now known as `Covidien'. 
The author wishes to thank M. Forina for providing data sets and guidance
regarding prior work on the topic.
Some of this research was conducted
at the time the author was a juror in the court of 
Judge Donald S. Mitchell, Department \#602, City and County of San Francisco,
California: ``Everyone's here and we. are. ready.''}}
\date{\today}

\begin{document}
\maketitle
\thispagestyle{fancy}
\begin{abstract}
\SRCEK{}, pronounced ``SIR check,'' is a technique for selecting useful
channels for affine modeling of a response by PLS.  The technique embeds the
discrete channel selection problem into the continuous space of predictor
preweighting, then employs a Quasi-Newton (or other) optimization algorithm to
optimize the preweighting vector.  Once the weighting vector has been
optimized, the magnitudes of the weights indicate the relative importance of
each channel.  The relative importances are used to construct $\cl$ different
models, the \kth{\ms} consisting of the $\ms$ most important channels.  The
different models are then compared by means of cross validation or an information 
criterion (\eg \BIC{}), allowing automatic selection of a good subset of the
channels.  The analytical Jacobian of the PLS regression vector with respect to
the predictor weighting is derived to facilitate optimization of the latter.
This formulation exploits the reduced rank of the predictor matrix to gain some
speedup when the number of observations is fewer than the number of predictors
(the usual case for \eg IR spectroscopy).  The method compares favourably with
predictor selection techniques surveyed by Forina \etal  \cite{Forinaetal2004}  
\end{abstract}
\normalsize\textbf{Keywords:~}\@keywords\vskip20pt
\nocite{USPAT8112375}

\section{Introduction}

The modern chemometrician suffers from an embarrassment of riches: 
investigative instruments (\eg NIR spectrometers) commonly allow measurements
in more discrete ``channels'' than the relevant underlying degrees of freedom or
the number of objects under investigation.  While a wide array of channels may
provide greater predictive power, some channels may confound prediction of the
relevant response, essentially measuring only ``noise.'' Moreover, the
principle of parsimony dictates that a predictive model must focus on as few of
the channels as practical, or fewer.

The problem was mitigated by the development of principal component
regression (PCR) and partial least squares (PLS)
\cite{dJSetal2004,El2004,RrKn2005,Gg1988,HLetal2003}, two algorithmic techniques
which essentially project many variate predictors onto a small dimensional
vector space before building a predictive model.  While these techniques are
very good at capturing the largest or most relevant underlying variations in
multivariate predictors, they are not impervious to disinformative or useless
predictors.  For example, it has been shown that the predictive ability of PLS is
degraded by a term quadratic in the number of channels.  \cite{NbCrr2005}

A number of clever and theoretically well grounded techniques exist for the
rejection of useless wavelengths. \cite{Chenetal2005,Hjaetal2003,Ntemetal2004,Bpjetal1999}
Rather than discuss them in any depth here,
I refer the reader to the excellent comparative study by Forina \etal
\cite{Forinaetal2004}  

\section{Problem Formulation}

Let \prd be an $\ro\cross\cl$ matrix of observed ``predictors'', with each
column of \prd a ``channel'' of the investigative instrument.  
Let \rsp be an $\ro$-dimensional column vector of
corresponding ``responses'' of each of the objects.  
Each row of \prd, with corresponding element of \rsp, represents an observation
of an ``object''.
In the general context of chemometrics, the number
of objects is generally far fewer than the number of channels:  $\ro\ll\cl$.

An affine predictive model consists of an
$\cl$-vector $\regb$ and scalar ``intercept'' \rego such that $\rsp \approx
\prd\regb + \rego\wunro,$ where \wunro is the vector of $\ro$ ones.  When calibrated on a
given collection of predictors and responses \prd and \rsp, different algorithms 
produce different affine models.  

A good model is marked by a residual with small norm.  That is,
$\sum_i \rese{i}^2$ is small, where $\resd = \rsp - \prd\regb - \rego\wunro$ is the 
residual.  However, a model which explains the observed data well may give poor
predictive ability over as-yet-unobserved objects, due to ``overfitting.''

Cross validation (CV) is used to address this deficiency.  The idea is to
divide the tested objects into two groups, one of which is used to build a
model \regmodel, the other is used to test the quality of the model.  By leaving
these validation or test objects out of the calibration, this technique
simulates the predictive ability of the observed data on unobserved data.

Because the division in two groups is arbitrary, the process is
often repeated in a round robin fashion, with different subdivisions.  The
quality of prediction of each division is then averaged.  In one
extreme form, known as ``leave one out'' (LOO) or ``delete-$1$'', the model 
building and testing is
performed $\ro$ times, each time with $\ro-1$ objects in the 
calibration group, and the single remaining object in the test group.
For comparing different subsets of the channels, Shao proved that delete-$1$ CV
is asymptotically inconsistent, \ie it has a nonzero probability of
overfitting. \cite{Shao:1993:LMS,Shao:1997:ATLMS}

Some terminology is now required.  The norm (more specifically the $2$-norm) of a 
vector \vect{v} is its Euclidian length:
\(\norm{\vect{v}} = \sqrt{\sum_i v_i^2} = \sqrt{\ip{\vect{v}}{\vect{v}}}.\)
The mean square norm of a $k$-dimensional vector, \vect{v} is
$\norm{\vect{v}}^2 / k$.  The mean squared error of a model \regmodel for given
data \prd and \rsp is the mean square norm of the residual $\rsp - \prd\regb -
\rego\wun{}.$ The mean squared error of cross validation (\MSECV) is the mean over each cross validation group of the mean square error of the model built by the calibration
group on the test group:
\[\MSECV = \oneby{J}\sum_{j=1}^J \oneby{\msz{j}}\norm{\rsp_j - \prd_j\regb_j -
{\rego}_j\wun{\msz{j}}}^2,\]
where $\prd_j$ and $\rsp_j$ are the predictors and responses of the \kth{j}
test group which contains \msz{j} objects, while \tuple{\regb_j,{\rego}_j} is built 
from the \kth{j} calibration group.  The mean squared error of prediction (\MSEP) of a
model is the mean square error of the data used to build the model:
\(\MSEP = \norm{\rsp - \prd \regb - \rego \wunro}^2/{\msz{}},\) where \regmodel is built
using (all) the data \prd and \rsp.

The prefix ``root'' refers to square root, thus the
root mean squared error of cross validation (\RMSECV) is $\sqrt{\MSECV}$;
similarly for \RMSEP, \etc


The goal of channel selection appears to be, given the observed
data, the CV groups and a model building technique, select the subset of the
$n$ available columns of
\prd which minimizes \RMSECV{} 
when only those
columns of \prd are used in the model building and testing.  In this
formulation the number of possible solutions is $2^\cl$; exhaustive search
becomes impractical when $\cl$ is larger than about 17.  Subset selection
heuristics such as Tabu search, Simulated Annealing (SA) and Genetic Algorithms
(GA) which generate and test subsets of the $\cl$ available channels can only hope
to explore a small part of the $2^\cl$ sized discrete space of possible subsets,
and are susceptible to backtracking and falling into local minima.  
\cite{Hjaetal2003,Ntemetal2004,Bpjetal1999}  Even when restricted
to subsets of no more than a given number of channels, say $d,$ 
heuristic search can cover only a small part of the search space, which has
size around $\cl^d$.

Considerable progress was made on the problem by the introduction of iterative
predictor weighting (IPW) by Forina \etal \cite{Forinaetal1999} This
channel selection technique reformulates the problem as one of selecting a
vector in $\cl$ dimensional Euclidian space, hereafter denoted \reals{n}, rather
than on the discrete space of binary $\cl$-dimensional 
vectors. 
In terms of size of the search space there would seem to be no
advantage to this reformulation.  However, I argue that the continuous
embedding allows the importance of each channel to be evaluated and changed
simultaneously, in parallel, rather than serially.   This will lead to a channel
selection technique with runtime asymptotically linear in the number of channels.

Predictive weighting can be viewed as a preprocessing step.  Let $\vlam \in
\reals{\cl}$ be a vector of weights\footnote{For this note, we ignore the
possibility of \vlam having a zero element.}.  Let \LAM be the diagonal
$\cl\cross\cl$ matrix whose diagonal is the vector \vlam.  Hereafter I will let
\diag{\vect{v}} denote the diagonal matrix whose diagonal is \vect{v}, so $\LAM
= \diag{\vlam}$.  The \vlam-weighted predictor matrix is the
product $\prd\LAM$:  the \kth{k} column of the weighted predictor matrix is the
\kth{k} column of \prd times \lame{k}.  Weighted predictors are then used in
the cross validation study, both in the calibration and test sets.  As such, the
quality of the cross validation (\RMSECV{})
can be viewed as a scalar
function of the vector \vlam, once the data and CV groups and model building
method (and order) have been fixed.

Note that when the regression model \regmodel is built by ordinary least squares
(OLS), the quality of cross validation is constant with respect to \vlam. This
occurs because the weighted regression model output by OLS is constant with
respect to nonzero predictor weighting, \ie $\LAM\regbp{\vlam}$, is constant over
all \vlam with nonzero elements.  PLS, however, does not share this
property, and the quality of cross validation is affected by predictor
weighting.  When used as a preprocessing technique prior to the application of
PLS, the usual strategy is to apply predictor weighting \vlam where $\lame{k} =
1/{\hat{\sigma}_k}^2,$ where ${\hat{\sigma}_k}^2$ is the sample standard
deviation of the \kth{k} channel of the predictors based on the observations 
in the entire sample
\prd, a technique called ``autoscaling.''  There is no reason to believe
\apriori that this choice of \vlam gives good cross validation.  Rather an
\apriori choice of weighting should depend on knowledge of the test instrument
or the tested objects.  Alternatively, one can let the data inform the choice
of predictor weighting.

The method I propose is to minimize the \RMSECV{} as a function of \vlam.  This
can be achieved by any of a number of modern optimization algorithms,
including BFGS \cite{NjWsj1999,liu89limited}, which I explore and advocate
here.  Once a local minimum has been found, the magnitude of the elements of
the optimal \vlam suggest the importance of the corresponding channels to the 
observed data.  This ordering suggest $\cl$ different models, the \kth{\ms}
consisting of the first $\ms$ channels by order of decreasing importance.  These
models can then be compared using any model selection technique, \eg
minimization of \RMSECV or an information criterion.
\cite{Shao:1993:LMS,Shao:1997:ATLMS,Ah:1974,Bk2004,CjeNaa1999}

I call this technique ``SRCEK{}'' (pronounced ``SIR check''), an
acronym for 
``Selecting Regressors by Continuous Embedding in K-dimensions,''
but also taken from the Slove word {\emph{sr\v{c}ek}},
meaning ``sweetheart.''  


\section{(W)PLS Regression}

The assumption underlying PLS regression is that predictor and response
take the form
\begin{align}
\Xk[]&=\uTk[0]\trans{\uPk[0]} + \uTk[1]\trans{\uPk[1]} + \ldots +
\uTk[l]\trans{\uPk[l]} + \Xe,\label{eqn:Xdeco}\\
\yk[]&=\uTk[0]{\uqk[0]} + \uTk[1]{\uqk[1]} + \ldots + \uTk[l]{\uqk[l]} +
\ye,\label{eqn:ydeco}
\end{align}
where the vectors \uTk are orthogonal, and the remainder terms, \Xe and \ye, are
random variables.  The vector \uTk[0] is \wun{}, the vector of all ones.
It is also assumed that the response remainder term, \ye, is homoscedastic, 
\ie
\[\E{\ye\trans{\ye}} = \sigmy\eye.\]
When this assumption cannot be supported, weighted PLS (WPLS) regression is
appropriate.  \cite{Hi1988,DBLP:journals/bioinformatics/FortL05,1998AIPC..430..253H}
Let \GAM be a symmetric positive definite matrix such that
$\GAM = \sclr{} \invs{\E{\ye\trans{\ye}}},$ for some (perhaps unknown) scalar
\sclr{}.
Weighted PLS assumes a decomposition of \Xk[] and \yk[] as in
equations~\ref{eqn:Xdeco} and \ref{eqn:ydeco}, but with the property that the 
vectors \uTk are \GAM-orthogonal: $\uTk[k]\GAM\uTk[j] = 0$ if $k\ne j$.

WPLS regression with $\fc$ factors computes $\ro \cross \fc$ matrix \Tmt,
$\cl\cross\fc$ matrix \Pmt and $\fc$ vector \qvc such that
\[\Xk[] \approx \Tk[0]\trans{\Pk[0]} + \Tmt \trans{\Pmt},\quad\text{and}\quad
\yk[] \approx \Tk[0]\qk[0] + \Tmt \qvc,\]
where \Tk[0] is the vector of all ones, and \trans{\Pk[0]} and \qk[0] are the
(\GAM-weighted) means of \Xk[] and \yk[].
The affine model constructed by WPLS takes the 
form $\regb = \Wmt \invs{\Parens{\trans{\Pmt}\Wmt}} \qvc,$ 
for some matrix \Wmt, with intercept $\rego = \qk[0] - \trans{\Pk[0]}\regb$.  
Thus
\[\Xk\regb + \Tk[0]\rego 
\approx \Tk[0]\trans{\Pk[0]}\regb + \Tmt\trans{\Pmt} \Wmt
\invs{\Parens{\trans{\Pmt}\Wmt}} \qvc + \Tk[0]\qk[0] - \Tk[0]\trans{\Pk[0]}\regb
= \Tk[0]\qk[0] + \Tmt\qvc \approx \yk[].\]



The use of the words ``weight'' or ``weighting'' in this context is
traditional, and parallels the usage for ordinary least squares.  It should not
be confused with the predictor weighting applied to the predictors.  To distinguish 
them, 
I will refer to \GAM{} as the response weights.  For the remainder of this
paper, I will assume that \GAM{} is a diagonal vector, $\GAM = \diag{\vgam}.$
This is not necessary for correctness of the algorithm, only for its fast
runtime, for which a sufficiently sparse \GAM would also
suffice.

\section{\RMSECV{} and its Gradient}

Given fixed data, I will presuppose the existence of a selected order, $\fc$.
The selection of $\fc$ should be informed by the underlying structure of the
objects under investigation, or by an automatic technique.
\cite{Forinaetal2004}  
In the chemometric context, the number of factors should be less (preferably
far less) than the number of objects: $\fc\le\ro$.
Let $\ffn{\vlam}$ be the \RMSECV{} for the CV when the CV groups are fixed, and
the affine model is built by $l$-factor WPLS with \vgam given, and using
predictor weights \vlam.  To employ quasi-Newton
minimization, the gradient of \ffn{} with respect to \vlam must be computed.
While this can be approximated numerically with little extra programmer effort,
the computational overhead can be prohibitive.  Thus I develop the analytic
formulation of the gradient.   At this point, the reader may wish to consult
the brief guide to vector calculus provided in \secref{bguide}.

In the general formulation, there is an response weight associated with each
observation.   These weights should be used in both the model construction
and error computation phases.  Thus, I rewrite the \RMSECV{} as a weighted
\RMSECV{} as follows:
\[\ffn{\vlam} = \sqrt{\oneby{J}\sum_{j=1}^J 
\frac{ \ip{\resdk[j]}{\GAMj[j]\resdk[j]}}{\ip{\Wun{j}}{\GAMj[j]\Wun{j}}} },\]
where $\resdk[j]$ is the residual of the \kth{j} test group, \GAMj[j] is
the diagonal matrix of the response weights of the \kth{j} test group, and
\Wun{j} is the appropriate length vector of all ones.  The gradient of
this is
\[\glam{\ffn{\vlam}} = \oneby{\ffn{\vlam}}\oneby{J}\sum_{j=1}^J
\oneby{\ip{\Wun{j}}{\GAMj[j]\Wun{j}}} \trans{\dlam{\resdk[j]}}\GAMj[j]\resdk[j],\]
Each residual takes the form
\[\resd = \ytst - \Parens{\Xtst\LAM\regbp{\vlam} + \wunt\rego\Parens{\vlam}},\]
thus the Jacobian of the residual is 
\begin{equation}
\dlam{\resd} = - \Xtst\Parens{\diag{\regbp{\vlam}} +
\LAM\dlam{\regbp{\vlam}}} - \wunt\trans{\glam{\rego}}.
\label{eqn:dresd}
\end{equation}
(Consult \eqnref{jprodrule} and \eqnref{diagrule} in \secref{bguide} for
proofs.)  Here and in the following, I use \Xcal, \ycal to refer to the
calibration objects, and \Xtst and \ytst to refer to the test objects of a
single cross validation partitioning.
This reduces the problem to the computation of the Jacobian and gradient of the WPLS
regression vector and intercept with respect to \vlam.  

%

\section{WPLS Computation}

A variant of the canonical WPLS computation is given in \algvref{vanillapls}.
This algorithm is different from the usual formulation in that the vectors 
\Wk, \Tk and \Pk are not normalized; it is simple to show, however, that the 
resultant vectors \Wk, \Tk and \Pk are identical to those produced by the 
canonical computation, \emph{up to scaling}.
The change in scaling does not affect the resultant regression vector,
$\regb$, nor does it change the matrix \Xk.

\begin{algorithm}[htb!]
\caption{Algorithm to compute the WPLS regression vector.\label{alg:vanillapls}}
\alginout{$\ro\cross\cl$ matrix and $\ro$ vector, number of
factors, and a diagonal response weight matrix.}{The regression vector and
intercept.}
\algname{WPLS}{$\Xk[0], \yk[], \fc, \GAM=\diag{\vgam}$}
\begin{algtab}
	$\Tk[0] \gets \wun{}$.\\
	\algforto{$k=0$}{$\fc$}
		$\tk \gets \trans{\Tk}\GAM\Tk.$\hfill(Requires \bigtheta{\ro} flops per loop.)\\
		$\Pk \gets \trans{{\Xk}} \GAM\Tk / \tk.$\hfill(Requires \bigtheta{\ro\cl} flops per
		loop.)\\
		$\qk \gets \trans{{\yk[]}} \GAM\Tk / \tk.$\hfill(Requires \bigtheta{\ro} flops per
		loop.)\\
		\algif{$\trans{\yk[]}\GAM\Tk = 0$}
				Maximum rank achieved; Let $\fc = k$ and break the for loop.\\
		\algend
		\algif{$k < \fc$}
			$\Xk[k+1] \gets \Xk - \Tk \trans{\Pk}.$\hfill(Requires \bigtheta{\ro\cl} flops
			per loop.)\label{algstep:PLSupX}\\
			$\Wk[k+1] \gets \trans{{\Xk[k+1]}}\GAM\yk[].$\hfill(Requires \bigtheta{\ro\cl} flops per loop.)\\
			$\Tk[k+1] \gets \Xk[k+1] \Wk[k+1].$\hfill(Requires \bigtheta{\ro\cl} flops per loop.)\\
		\algend
	\algend
	Let \Wmt be the matrix with columns $\Wk[{\onetox{l}}].$  Similarly
	define \Mtx{\Psym}, \vect{q}.\\
	$\regb \gets \Wmt \invs{\Parens{\trans{\Pmt}\Wmt}} \qvc.$
		(Requires \bigtheta{\cl\fc} flops, using back substitution.)\\
	$\rego \gets \qk[0] - \trans{\Pk[0]}\regb.$\\
\algreturn \tuple{\regb, \rego}
\end{algtab}
\end{algorithm}

I prove some properties of \algref{vanillapls}, nearly all of which hold for the
canonical WPLS algorithm:
\begin{lemma}\label{lem:vplsprops}
Let \Xk, \GAM, \Wk, \Tk, and \Pk be as in \algref{vanillapls}, then:
\begin{compactenum}
\item $\ip{\Wk[k]}{\Pk[k]} = 1$, for $k\ge1$.
\label{item:wkpk}
\item $\Xk[k+j]\Wk[k] = \vz$ for all $j \ge 1$, and $k\ge1$.
\label{item:xwk}
\item $\ip{\Wk[k+j]}{\Wk[k]} = 0$ and $\ip{\Pk[k+j]}{\Wk[k]} = 0$ for all $j
\ge 1$, and $k\ge1$.
\label{item:wkwk}
\item $\trans{\Xk[k+j]}\GAM\Tk[k] = \vz$ for all $j \ge 1$, and $k\ge0$.
\label{item:xtk}
\item $\ip{\Tk[k+j]}{\GAM\Tk[k]} = 0$ for all $j \ge 1$, and $k\ge0$.
\label{item:tktk}
\item $\Wk[k+1] = \Wk[k] - \Pk[k]\ip{\Tk[k]}{\GAM\yk[]}$ for $k\ge1$, and thus
$\Pk[k] \in \operatorname{span}\sngtn{\Wk[k],\Wk[k+1]}.$
\label{item:pkspan}
\item 
$\ip{\Wk[k+j]}{\Pk[k]} = 0$, and
$\ip{\Pk[k+j]}{\Pk[k]} = 0$
for all $j > 1$, and $k\ge1$.  (Note the strict inequality
for $j$.)
\label{item:wklpk}
\end{compactenum}
\end{lemma}
\begin{proof}
First note that the update of \Xk is given by 
\[\Xk[k+1] = \Xk - \Tk \trans{\Pk} =
\Parens{\eye - \frac{\Tk\trans{\Tk}\GAM}{\tk}} \Xk = \Mk\Xk.\]
It can also be expressed as 
\[\Xk[k+1] = \Xk - \Tk \trans{\Pk} = \Xk \Parens{\eye - \Wk\trans{\Pk}} =
\Xk\Lk,\]
for $k>0$.
Now the parts of the lemma:
\begin{compactenum}
\item[\bf \ref{item:wkpk}]%
\[\ip{\Wk}{\Pk} = \frac{\trans{\Wk}\trans{\Xk}\GAM\Tk}{\tk} =
\frac{\trans{\Tk}{\GAM\Tk}}{\tk} = \frac{\tk}{\tk} = 1.\]
\item[\bf \ref{item:xwk}] 
First I prove $\Xk[k+1]\Wk = \vz$ as follows:
\begin{align*}
\Xk[k+1]\Wk &= \Xk \Parens{\eye - \Wk\trans{\Pk}}\Wk\\
&=\Xk \Parens{\Wk - \Wk \ip{\Pk}{\Wk}} = \Xk \Parens{\Wk - \Wk} = \vz,
\end{align*} which follows because
$\ip{\Pk}{\Wk} = 1$
To prove for general $j,$ we note that 
$\Xk[k+j] = \Mk[k+j-1]\Xk[k+j-1] = \Mk[k+j-1]\Mk[k+j-2]\Xk[k+j-2] = \ldots =
\Mtx{M} \Xk[k+1]$, so 
$\Xk[k+j]\Wk = \Mtx{M}\Xk[k+1]\Wk = \Mtx{M}\vz = \vz.$
\item[\bf \ref{item:wkwk}] This is trivial because, for example, 
$\ip{\Wk[k+j]}{\Wk} = \trans{\yk[]}\GAM\Xk[k+j]\Wk = \trans{\yk[]}\GAM\vz$, following
from the previous part.  Similarly for \ip{\Pk[k+j]}{\Wk}.
\item[\bf \ref{item:xtk}] Again I prove for $j=1$ first:
\begin{align*}
\trans{\Xk[k+1]}\GAM\Tk &= \trans{\Xk} \trans{\Parens{\eye -
\frac{\Tk\trans{\Tk}\GAM}{\tk}}}\GAM\Tk\\
&= \trans{\Xk} \Parens{\GAM\Tk - \frac{\GAM\Tk\trans{\Tk}\GAM\Tk}{\tk}}\\
&= \trans{\Xk} \GAM\Parens{\Tk - \frac{\Tk\tk}{\tk}}
= \trans{\Xk} \GAM\Parens{\Tk - \Tk} = \vz.
\end{align*}
Similarly to above we can prove for $j > 1$ by writing $\Xk[k+j] = \Xk[k+1]\Mtx{L}$.
\item[\bf \ref{item:tktk}] By the previous result, $\ip{\Tk[k+j]}{\GAM\Tk} =
\trans{\Wk[k+j]}\trans{\Xk[k+j]}\GAM\Tk = 
\trans{\Wk[k+j]}\vz = 0$.
\item[\bf \ref{item:pkspan}] 
This is by simple definition:
\[\Wk[k+1] = \trans{\Xk[k+1]}\GAM\yk[] = \trans{\Parens{\Xk[k] -
\Tk[k]\trans{\Pk[k]}}}\GAM\yk[] = \Wk[k] - \Pk[k]\trans{\Tk[k]}\GAM\yk[].\]
Thus by rearranging, \Pk[k] is a linear combination of \Wk[k] and \Wk[k+1].
%

\item[\bf \ref{item:wklpk}] 
From \itemref{pkspan}: 
$\ip{\Wk[k+j]}{\Pk} = \ip{\Wk[k+j]}{\Parens{\sclr[1]\Wk[k] + \sclr[2]\Wk[k+1]}}.$
Since $j>1$, by orthogonality of the \Wk[] (\itemref{wkwk}), the right hand
side is zero, as needed.

For the \Pk[], first use \itemref{wkwk} to assert
$0 = \ip{\Pk[k+j]}{\Wk[k+1]},$ then rewrite \Xk[k+1] in \Wk[k+1]:
\begin{align*}
0 &= \ip{\Pk[k+j]}{\Wk[k+1]} 
= \ip{\Pk[k+j]}{\Bracks{\trans{\Parens{\Xk - \Tk\trans{\Pk}}}\GAM\yk[]}}\\
0 &= \ip{\Pk[k+j]}{\Bracks{\Wk - \Pk\trans{\Tk}\GAM\yk[]}}
= \ip{\Pk[k+j]}{\Wk} - \ip{\Pk[k+j]}{\Pk\trans{\Tk}\GAM\yk[]}\\
0 &= 0 - \Parens{\ip{\Pk[k+j]}{\Pk}}\trans{\Tk}\GAM\yk[],
\end{align*}
and thus either $\ip{\Pk[k+j]}{\Pk} = 0$ as desired or 
$\trans{\Tk}\GAM\yk[] = 0$.  The algorithm detects this possibility and terminates
if it holds.

\end{compactenum}
\end{proof}

Thus, as in canonical WPLS, the matrix $\trans{\Pmt}\Wmt$ is bidiagonal upper
triangular; however, for this variant, the matrix has unit main diagonal.
This variant of the algorithm is more amenable to a Jacobian computation,
although conceivably it could be susceptible to numerical underflow or
overflow.

\section{A Rank Sensitive Implicit WPLS Algorithm}

I now present a different way of computing the same results as
\algref{vanillapls}, but by reusing old computations to compute seemingly
unnecessary intermediate quantities which will be useful in the Jacobian computation.  
Moreover, the Jacobian computation will exploit the assumption that $\ro\ll\cl$ to gain 
an asymptotic reduction in runtime.  This is achieved by performing the
computations in the $\ro$-dimensional space, that is in the quantities related
to \yk[] and \Tk, and avoiding the $\cl$-dimensional quantities \Wk and \Pk.

The variant algorithm introduces the ``preimage'' vectors \Vk and the preimage
of the regression coefficient $\rega$.  By preimage, I mean in
$\trans{\Xk[0]}\GAM$, thus, as
an invariant, these vectors will satisfy $\Wk = \trans{\Xk[0]}\GAM\Vk$ and $\regb =
\trans{\Xk[0]}\GAM\rega.$
The variant algorithm also computes the vectors $\Qk[k] =
{\Xk[0]\trans{\Xk[0]}\GAM} \Vk[k]$, and $\Uk[k] = \Qk[k] -
\Tk[0]\trans{\Tk[0]}\GAM\Qk[k]/\tk[0]$, and the
scalars 
$\rk = \ip{\yk[]}\GAM{\Tk}$ and $\wk = \trans{\Tk}\GAM\Uk[k+1]/\tk$.
%

Note that
any explicit updating of the matrix \Xk is absent from this version of the algorithm,
rather the updating is performed implicitly.  This will facilitate the
computation of the Jacobian when \Xk[0] is replaced in the sequel by
$\Xk[0]\LAM$.  The following lemma confirms that this variant form of the
algorithm produces the same results as \algref{vanillapls}, that is the same
vectors \Tk[] and \vect{\qk[]}, consistent vectors \Vk[], and produces the same
affine model \regmodel.

\begin{lemma}\label{lem:oddplsprops}
Let \Xk, \GAM, \Wk, \Tk, \Pk, and \qk be as in \algref{vanillapls}.  
Let \Vk be the preimage of \Wk, \ie \Vk is a vector such that $\Wk =
\trans{\Xk[0]}\GAM\Vk$.  
Let $\Qk[k] = {\Xk[0]\trans{\Xk[0]}\GAM} \Vk[k]$, and 
$\Uk[k] = \Qk[k] - \Tk[0]\trans{\Tk[0]}\GAM\Qk[k]/\tk[0]$ for $k\ge1$.
Let
$\rk = \ip{\yk[]}\GAM{\Tk}$, and $\wk = \trans{\Tk}\GAM\Uk[k+1]/\tk$ for
$k\ge0.$
%
Then the following hold:
\begin{compactenum}
\item $\Pk = \trans{\Xk[0]}\GAM\Tk / \tk,$ for $k\ge0$.
\label{item:repk}
\item $\wk = \ip{\Pk[k]}{\Wk[k+1]}$ for $k\ge1$, and $\wk[0] = 0$.
\label{item:defwk}
\item $\Tk[k+1] = \Uk[k+1] - \wk[k] \Tk[k]$ for $k \ge 0.$
\label{item:retk}
\item 
$\Vk[k+1] = \Vk[k] - \qk[k]\Tk[k]$ for $k\ge0$,
where $\Vk[0] = \yk[]$ by convention.
\label{item:revk}
\end{compactenum}
\end{lemma}
\begin{proof}
The parts of the lemma:
\begin{compactenum}

\item[\bf \ref{item:repk}]%
Rewrite \Xk, then use the fact that the \Tk[] are \GAM-orthogonal 
(\lemref{vplsprops} \itemref{tktk}):
\begin{align*}
\Pk &= \trans{\Xk}\GAM\Tk / \tk = 
	\trans{\Parens{\Xk[0] - \Tk[0]\trans{\Pk[0]} - \ldots 
	- \Tk[k-1]\trans{\Pk[k-1]}}}\GAM\Tk / \tk\\
 &= \Parens{\trans{\Xk[0]} - \Pk[0]\trans{\Tk[0]} - \ldots
	- \Pk[k-1]\trans{\Tk[k-1]}}\GAM\Tk / \tk\\
 &= \trans{\Xk[0]}\GAM\Tk / \tk.
\end{align*}

\item[\bf \ref{item:defwk}]%
By definition, and \GAM-orthogonality of the \Tk[]:
\begin{align*}
\wk &= \trans{\Tk}\GAM\Uk[k+1]/\tk 
	= \trans{\Tk}\GAM\Parens{\Qk[k+1] - \Tk[0]\sclr{}}/\tk 
	= \trans{\Tk}\GAM\Qk[k+1]/\tk,\\
 &= \trans{\Tk}\GAM \Xk[0]\trans{\Xk[0]}\GAM \Vk[k+1]/\tk
  = \Parens{\trans{\Tk}\GAM \Xk[0]/\tk}\Parens{\trans{\Xk[0]}\GAM \Vk[k+1]}\\
 &= \trans{\Pk}\Wk[k+1],\quad\text{using \itemref{repk}.}
\end{align*}
To show $\wk[0] = 0$ it suffices to note that \Uk is chosen to be
\GAM-orthogonal to \Tk[0]:
\[\trans{\Tk[0]}\GAM\Uk 
= \trans{\Tk[0]}\GAM\Parens{\Qk - \Tk[0]\trans{\Tk[0]}\GAM\Qk/\tk[0]}
= \trans{\Tk[0]}\GAM\Qk - \Parens{\tk[0]/\tk[0]}\trans{\Tk[0]}\GAM\Qk.
\]
\item[\bf \ref{item:retk}]%
For $k>0$, rewrite \Xk[k+1]:
\begin{align*}%
\Tk[k+1]&= \Xk[k+1] \Wk[k+1]\\
&= \Parens{\Xk[0] - \Tk[0]\trans{\Pk[0]} - \ldots - \Tk[k-1]\trans{\Pk[k-1]} -
\Tk[k]\trans{\Pk[k]}} \Wk[k+1],\\
&= \Xk[0]\Wk[k+1] - \Tk[0]\trans{\Pk[0]}\Wk[k+1] - \Tk[k]\trans{\Pk[k]} \Wk[k+1],
	\quad\text{(\lemref{vplsprops} \itemref{wklpk}),}\\
&= \Qk[k+1] - \Tk[0]\trans{\Tk[0]}\Xk[0]\Wk[k+1]/\tk[0] - \Tk[k]\wk[k],
\quad\text{(using \itemref{defwk}),}\\
&= \Uk[k+1] - \Tk[k]\wk[k].
\end{align*}
For $k=0$, since $\wk[0] = 0,$ it suffices to show $\Tk[1] = \Uk[1].$
Rewriting \Xk[1]:
\begin{align*}%
\Tk[1]&= \Xk[1] \Wk[1] = \Xk[0]\Wk[1] - \Tk[0]\trans{\Pk[0]}\Wk[1]\\
&= \Qk[1] - \Tk[0]\trans{\Tk[0]}\Xk[0]\Wk[1]/\tk[0]
= \Uk[1].
\end{align*}
\item[\bf \ref{item:revk}]%
First, for $k > 0,$ restate \itemref{pkspan} of \lemref{vplsprops}:
$\Wk[k+1] = \Wk[k] - \Pk[k]\trans{\Tk[k]}\GAM\yk[]$.  Factoring to
preimages using \itemref{repk} gives
$\Vk[k+1] = \Vk[k] - \Tk[k]\trans{\Tk[k]}\GAM\yk[]/\tk[k]$.  
The definition of \qk[k] then gives the result.

For $k=0,$ rewrite \Xk[1]:
\begin{align*}
\Wk[1] &= \trans{\Parens{\Xk[0] - \Tk[0]\trans{\Pk[0]}}}\GAM\yk[]
= \trans{\Xk[0]}\GAM\yk[] -
\trans{\Xk[0]}\GAM\Tk[0]\Tk[0]\GAM\yk[]/\tk[0],\quad\text{thus}\\
\Vk[1] &= \yk[] - \Tk[0]\qk[0].
\end{align*}%
\end{compactenum}
\end{proof}
For concreteness, I present the WPLS algorithm via intermediate calculations as
\algvref{oddpls}.
\begin{algorithm}[htb!]
\caption{Algorithm to compute the WPLS regression vector, with \Xk implicit.\label{alg:oddpls}}
\alginout{Matrix and vector, factors, and diagonal response weight matrix.}{The
regression vector and intercept.}
\algname{implicitWPLS}{$\Xk[0], \yk[], \fc, \GAM=\diag{\vgam}$}
\begin{algtab}
	$\Mij{}{} \gets \eye[\fc],$ the $\fc\cross\fc$ identity matrix.\\
	Precompute $\Xk[0]\trans{\Xk[0]}\GAM$.\\
	$\Tk[0] \gets \wun{}, \Vk[0] \gets \yk[]$.\\
	\algforto{$k=0$}{$\fc$}
		$\rk \gets \trans{\yk[]}\GAM\Tk[k].$\\
		\algif{$\rk = 0$}
			Full rank achieved; Let $\fc = k$ and break the for loop.\\
		\algend
		$\tk \gets \trans{\Tk}\GAM\Tk.$\\
		$\qk \gets \rk / \tk.$\\
		\algif{$k < \fc$}
			$\Vk[k+1] \gets \Vk[k] - \qk[k]\Tk[k].$\\
			Let $\Qk[k+1] \gets \Parens{\Xk[0]\trans{\Xk[0]}\GAM} \Vk[k+1]$.\\
			Let $\Uk[k+1] \gets \Qk[k+1] - \Tk[0]\trans{\Tk[0]}\GAM\Qk[k+1]/\tk[0]$.\\
			$\wk[k] \gets \trans{\Tk[k]} \GAM \Uk[k+1] / \tk[k]$.\\
			$\Tk[k+1] \gets \Uk[k+1] - \wk[k]\Tk[k]$.\\
			\algif{$k > 0$}
				$\Mij{k}{k+1} \gets \wk[k]$.\\
			\algend
		\algend
	\algend
	Let \Mtx{V} be the matrix with columns $\Vk[{\onetox{\fc}}].$  Similarly
	define \vect{q}.\\
	$\rega \gets \Mtx{V} \invs{\Mij{}{}} \vect{q},$
	$\regb \gets \trans{\Xk[0]}\GAM\rega,$
	$\rego \gets \qk[0] - \trans{\Tk[0]}\GAM\Xk[0] \regb / \tk[0].$\\
\algreturn \tuple{\regb, \rego}.
\end{algtab}
\end{algorithm}

\section{WPLS Computation with Jacobian}

Now I amend \algref{oddpls} with derivative computations to create an 
algorithm that computes the
regression coefficient for input $\Xk[1]\LAM,$ and \yk[1], and returns the
preimage of the regression vector, $\rega,$ as well as its Jacobian \drega, and
the gradient of the intercept, \glam{\rego}.
This is given as \algvref{fullonpls}.  In addition to the intermediate quantities
used in \algref{oddpls}, this algorithm also computes some intermediate
derivatives, some of which need to be stored until the end of the computation.
The required derivatives are \dlam{\Vk[k]}, \glam{\qk[k]} and \glam{\wk[k]} for
$k\ge1$,  and \glam{\rk[k]}, \dlam{\Uk[k]}, \glam{\tk[k]}, and \dlam{\Tk[k]} for the
most recent $k$.

\begin{algorithm}[htbp!]
\caption{Algorithm to compute the WPLS regression vector and 
Jacobian.\label{alg:fullonpls}}
\alginout{Predictor and response, factors, response weights and predictor 
weights.}{The preimage of the regression vector and its Jacobian, and the
intercept and its gradient.}
\algname{WPLSandJacobian}{$\Xk[1], \yk[], \fc, \vgam, \vlam$}
\begin{algtab}
	Precompute $\XLXG$.\\
	$\Tk[0] \gets \wun{}, \dlam{\Tk[0]} \gets \Mz, \Vk[0] \gets \yk[],
	\dlam{\Vk[0]} \gets \Mz$.\\
	\algforto{$k=0$}{$\fc$}
		$\rk \gets \trans{\yk[]}\GAM\Tk[k], \glam{\rk} \gets
		\trans{\dlam{\Tk[k]}}\GAM\yk[].$\\
		\algif{$\rk = 0$}
			Full rank achieved; Let $\fc = k$ and break the for loop.\\
		\algend
		$\tk \gets \trans{\Tk}\GAM\Tk, \glam{\tk} \gets 2 \trans{\dlam{\Tk[k]}}
		\GAM\Tk[k].$\\
		$\qk \gets \rk / \tk, \glam{\qk} \gets \wrapparens{\tk\glam{\rk} -
		\rk\glam{\tk}}/\tk^2.$\\
		\algif{$k < \fc$}
			$\Vk[k+1] \gets \Vk[k] - \qk[k]\Tk[k],$
			\newline $\dlam{\Vk[k+1]} \gets \dlam{\Vk[k]} - \qk[k]\dlam{\Tk[k]} -
			\Tk[k]\trans{\glam{\qk[k]}}.$\\
			Let $\Qk[k+1] \gets \Parens{\XLXG} \Vk[k+1]$,\newline
			$\dlam{\Qk[k+1]} \gets \XLXG \dlam{\Vk[k+1]} + 2 \Xk[0] \LAM
			\diag{\trans{\Xk[0]}\GAM\Vk[k+1]}.$\\
			Let $\Uk[k+1] \gets \Qk[k+1] - \Tk[0]\trans{\Tk[0]}\GAM\Qk[k+1]/\tk[0]$,\newline
			$\dlam{\Uk[k+1]} \gets \dlam{\Qk[k+1]} -
			\Tk[0]\trans{\Tk[0]}\GAM\dlam{\Qk[k+1]}/\tk[0]$.\\
			$\wk[k] \gets \trans{\Tk[k]} \GAM \Uk[k+1] / \tk[k]$,
			\newline $\glam{\wk[k]} \gets
			\Parens{\trans{\dlam{\Tk[k]}}\GAM\Uk[k+1] +
			\trans{\dlam{\Uk[k+1]}}\GAM\Tk[k] - \wk[k] \glam{\tk[k]}} / \tk$.\\
			$\Tk[k+1] \gets \Uk[k+1] - \wk[k]\Tk[k]$,
			\newline $\dlam{\Tk[k+1]} \gets \dlam{\Uk[k+1]} - \wk[k]\dlam{\Tk[k]} -
			\Tk[k]\trans{\glam{\wk[k]}}.$\\
		\algend
	\algend
	Let $\vk[l] \gets \qk[l]$, $\glam{\vk[l]} \gets \glam{\qk[l]}$.\\
	\algforto{$\jx=l-1$}{$1$}
		$\vk[\jx] \gets \qk[\jx] - \wk[\jx] \vk[\jx+1]$,\newline
		$\glam{\vk[\jx]} \gets \glam{\qk[\jx]} - \wk[\jx] \glam{\vk[\jx+1]} - \vk[\jx+1]
		\glam{\wk[\jx]}$.\\
	\algend
	$\rega \gets \Mtx{\Vsym} \vect{\vk[]},$
	$\drega \gets \Mtx{\Vsym} \dlam{\vect{\vk[]}}.$\\
	\algforto{$\jx=1$}{$l$}
		\label{algstep:FOfaketensor}
		$\drega \gets \drega + \vk[\jx] \dlam{\Vk[\jx]}.$\\
	\algend
	$\rego \gets \trans{\Tk[0]}\GAM\Parens{\yk[] - \XLXG \rega} / \tk[0]$,
	\newline $\dlam{\rego} \gets -\trans{\Parens{ \XLXG \drega + 2
	\Xk[0]\LAM\diag{\trans{\Xk[0]}\GAM\rega} }} \GAM\Tk[0]/\tk[0].$\\

\algreturn \tuple{\rega,\drega,\rego,\glam{\rego}}.
\end{algtab}
\end{algorithm}

Careful inspection and the vector calculus rules outlined in \secref{bguide}
are all that is required to verify that \algref{fullonpls} correctly computes
the Jacobian of the model $\regb$.  The only theoretical complication in the
transformation of \algref{oddpls} to \algref{fullonpls} is the explicit formulation
of the back substitution to compute $\vect{\vk[]} = \invs{\Mij{}{}} \vect{q}.$
Given that \Mij{}{} is upper triangular, bidiagonal with unit diagonal,
inspection reveals that the back substitution in \algref{fullonpls} is computed
correctly.

Inspection of \algref{vanillapls} reveals that WPLS computation requires
\bigtheta{\ro\cl\fc} floating point operations, where \Xk[1] is $\ro
\cross\cl$, and $\fc$ is the ultimate number of factors used. Thus a numerical
approximation to the Jacobian using $\cl$ evaluations of \algref{vanillapls}
gives an algorithm with asymptotic runtime of \bigtheta{\ro\cl^2\fc}.
Inspection of \algref{fullonpls} reveals that it computes the Jacobian exactly
in \bigtheta{\ro^2\cl\fc}.   The runtime limiting operation is the 
multiplication $\Parens{\XLXG}\dlam{\Vk[k+1]}$ in the calculation of
\dlam{\Uk[k+1]}, with runtime of \bigtheta{\ro^2\cl} per loop.

It would appear that one would incur a further cost of \bigtheta{\ro\cl^2} in
the conversion of \drega to \dregb, as it requires the multiplication
$\LAM\trans{\Xk[1]}\GAM\drega$.  However, this can be avoided if the ultimate goal
is computation of the Jacobian of the residual, rather than the Jacobian of the
regression coefficients.  Referring back to \eqnref{dresd}
, we have
\begin{align*}
\dlam{\resd} + \wunt\trans{\glam{\rego}}
&= - \Xtst\Parens{\diag{\regbp{\vlam}} +
\LAM\dlam{\regbp{\vlam}}},\\
&= - \Xtst\Parens{\diag{\LAM\trans{\Xcal}\GAM\regap{\vlam}} +
	\LAM\dlam{\LAM\trans{\Xcal}\GAM\regap{\vlam}}},\\
&= - \Xtst\Parens{\diag{\LAM\trans{\Xcal}\GAM\regap{\vlam}} +
	\LAM\diag{\trans{\Xcal}\GAM\regap{\vlam}} +
	\LAMk2\trans{\Xcal}\GAM\dlam{\regap{\vlam}}},\\
&= - 2\Xtst\diag{\LAM\trans{\Xcal}\GAM\regap{\vlam}} -
	\Parens{\Xtst\LAMk2\trans{\Xcal}}\GAM\dlam{\regap{\vlam}}.
\end{align*}
Letting \rotst be the number of objects in the test group, the
multiplication 
$\Xtst\LAMk2\trans{\Xcal}$ requires \bigo{\ro\rotst\cl} flops, and the
multiplication $\Parens{\Xtst\LAMk2\trans{\Xcal}}\GAM\dlam{\regap{\vlam}}$ also 
requires \bigo{\ro\rotst\cl} flops.  Thus the computation of \dlam{\resd} can be
done with \bigo{\ro\rotst\cl + \ro^2\cl\fc} flops, which is linear in $\cl$.

For concreteness, the residual computation with analytic Jacobian was coded and
compared for accuracy and speed against a ``slow'' analytic version (one which 
does not exploit the reduced rank in the Jacobian computation) and a
numerical approximation to the Jacobian.  Run times are compared in
\figref{comprtimes} for varying number of channels; the difference in
asymptotic behavior with respect to $\cl$ is evident.
For the case of $40$ calibration objects and $10$ test objects generated
randomly with $2000$ channels, the fast analytic computation of residual
Jacobian took about $1.7$ seconds, the slow analytic took about $44$ seconds, and the
numeric approximation took about $84$ seconds on the platform tested (see
\secref{impnotes} for details).
Note that the ``slow'' analytic version is actually preferred in the case that
$\ro \ge \cl,$ as it runs in time \bigtheta{\ro\cl^2\fc}.  However, in
spectroscopy it is usually the case that $\ro\ll\cl$.


\begin{figure}[htb!]
\centering
	\psfrag{na}[rb][rb]{fast a}
	\psfrag{sa}[rb][rb]{slow a}
	\psfrag{nu}[rb][rb]{numeric}
	\psfrag{n}[][][2]{$\cl$}
	\psfrag{t}[][][2][0]{time (secs)}
	\includegraphics[angle=270,width=.85\columnwidth]{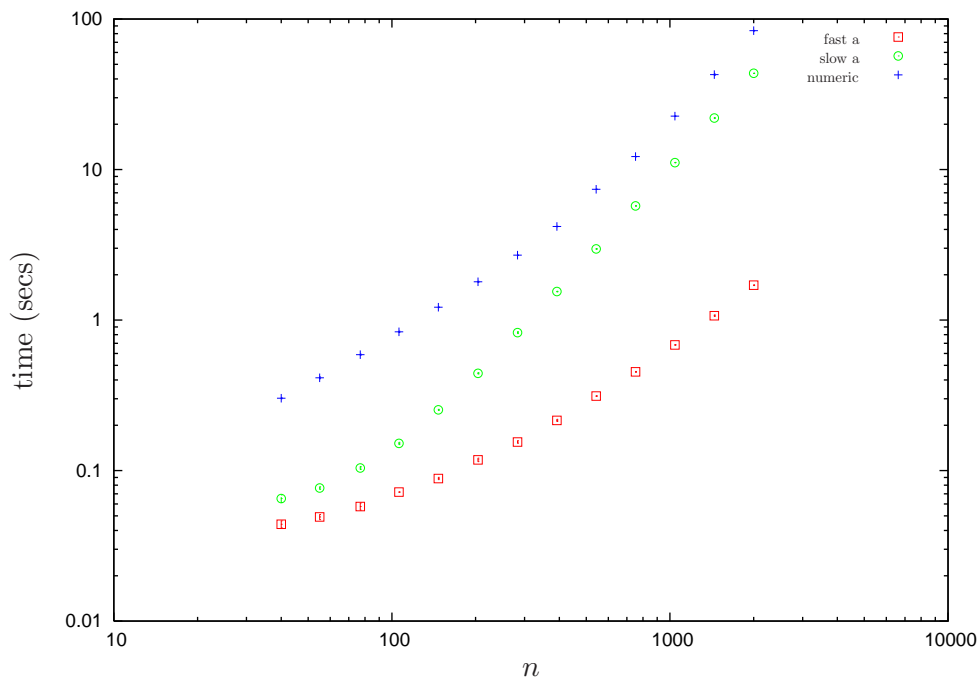}
\caption{Run times of the fast analytical computation of the
Jacobian of the residual are compared against a slow analytic and a numerical
approximation.  The number of channels, $\cl$ is shown in the horizontal axis,
while the vertical axis is CPU time in seconds.  The number of objects in the
calibration and test groups remained constant, at $40$ and $10$ throughout, as
did the number of PLS factors, $5$.  Times are the mean of seven runs, and
single standard deviation bars are plotted, although they are mostly too small
to see.  See \secref{impnotes} for details on the tested platform.}\label{fig:comprtimes} 
\end{figure}

\section{The BFGS Algorithm}

The BFGS algorithm\footnote{Named for its discoverers, Broyden, Fletcher,
Goldfarb and Shanno.} is a quasi-Newton optimization algorithm.  That is, the algorithm
models a scalar function of many variables by a quadratic function with an
approximate Hessian.  The approximation to the Hessian is improved at each step
by two rank one updates.  The BFGS algorithm enjoys a number of properties which
make it attractive to the numerical analyst:  provable superlinear global
convergence for some convex optimization problems;  provable superlinear local 
convergence for some nonconvex problems;  robustness and good performance in
practice; deterministic formulation; relative simplicity of implementation;
and, perhaps most importantly to the practical analyst, the algorithm has been
implemented in a number of widely available libraries and packages, many of
which accept the objective function as a blackbox.
\cite{NjWsj1999,nocedal91theory}

The BFGS algorithm is an iterative solver.  That is, it starts with some
initial estimate of a good \vlam, say \vlamk[0], and produces successive
estimates, \vlamk, which are supposed to converge to a local minimizer of the
objective function.  Each iteration consists of a computation of the gradient
of the objective at \vlamk.  The algorithm constructs a search direction,
call it \vrok, by multiplying the inverse approximate Hessian by the negative gradient.
Then a line search is performed to find an acceptable step in the search
direction, that is to find the \salk used to construct
$\vlamk[k+1] = \vlamk[k] + \salk[k] \vrok[k].$
In the backtracking algorithm used to perform line search described by Nocedal
and Wright, a number of prospective values of \salk may be tested; the
objective function must be computed for each prospective value, but the
gradient need not be computed. \cite[Algorithm 3.1]{NjWsj1999}  A fast
implementation of the BFGS algorithm should not query the blackbox function for
gradients during the backtracking phase.
\nocite{MjjTdj1992}

As mentioned above, the BFGS requires some initial estimate of the Hessian of
the objective function.  When a good initial estimate of the Hessian is
impractical, the practical analyst punts, and resorts to the identity matrix.
Under this choice, the first search direction is the negative gradient, \ie the
direction of steepest descent.  The BFGS constructs better estimates of the
Hessian by local measurement of the curvature of the objective function.

Depending on the implementation, the BFGS algorithm may have to store the
approximate Hessian of the objective function or the inverse approximate
Hessian.  In either case, the storage requirement is \omeg{\cl^2}.   To avoid
this, one can use the limited memory BFGS algorithm, which approximates the
Hessian by a fixed number of the previous iterative updates, which avoids the
need for quadratic storage.  This method evidently works as well as BFGS in
practice for many problems.  \cite{liu89limited,gill97limitedmemory,nocedal91theory}


\section{Selecting Wavelengths from an Optimal \vlam}
\label{sec:unembed}

Once a predictor weighting \vlam has been found which gives a small \RMSECV{},
one must use the \vlam to select a subset of the channels.  That is, one must
reverse the embedding, finding a subset of the channels in the discrete space
of all such subsets which somehow approximates the continuous solution given by
\vlam.  Without loss of generality, one may assume that \vlam has unit norm, \ie
$\norm{\vlam} = 1$, since the effective WPLS
regression vector is invariant under scaling, \ie 
$\sclr\LAM\regbp{\sclr \vlam}$ is constant for all nonzero values of
$\sclr.$  This latter fact is proved by considering the output of the
canonical WPLS algorithm, which normalizes the vectors \Wk and \Tk.
Moreover, I assume that the elements of \vlam are nonnegative, again without
loss of generality.

Clearly, the weightings of the channels somehow signify their importance, and
can be used in the selection of a subset of the channels. 
The ordering in significance indicated by \vlam suggests $\cl$ different possible
choices of subsets\footnote{To be fair, the trivial model, which estimates the
\yk[] values as constant, should also be considered.}, the \kth{\ms} of which is
the subset with the $\ms$ most significant channels.  If the acceptable number of
channels is bounded by an external restriction, say an upper bound of $\cl_f,$ then 
one should select the subset of the $\cl_f$ most significant channels.
Without any external restrictions, one should select the subset of channels (or
``model'') which minimizes some measure of predictive quality, such as \RMSECV
or an information criterion like Schwarz' Bayesian Information Criterion
(\BIC).
%

The asymptotic (in $\ro$) consistency of model selection criteria was examined by
Shao. \cite{Shao:1993:LMS,Shao:1997:ATLMS}  A number of differences exist
between the formulation studied by Shao and that presented here: our design
matrix is assumed to be of reduced rank (\ie \eqnref{Xdeco} describes a reduced
rank matrix) and non-deterministic\footnote{Shao dismisses this complication by
stating his results will hold almost surely under certain conditions.}; our
affine model is built by PLS rather than OLS. However, absent any extant
results for the reduced rank formulation, I follow Shao's work, which you may
take with a grain of salt.

I will focus on two model comparison criteria suggested by Shao: delete-$d$ CV,
and \BIC.  Delete-$d$ CV is regular cross validation with $d$ objects in the
validation set.  The model which minimizes \RMSECV{} under the given grouping
is selected.  Because \nchoosek{\ro}{d} can be very large, only a number of
the possible CV groupings are used.  Shao's study suggests that Monte Carlo
selection of the CV groups can be effective with only \bigo{\ro} of the
possible groupings used.  Shao also proved that $d/\ro \to 1$ is a
prerequisite for asymptotic consistency.  In his simulation study, he used $d
\approx \ro - \ro^{3/4}$, and found that it outperformed delete-$1$ CV,
especially in those tests where selecting overly large models is possible.
\cite{Shao:1993:LMS}



Shao also examines a class of model selection criteria which contains the
General Information Criterion described by Rao and Wu, the minimization
of which, under certain assumptions, is equivalent to minimizing \BIC.
\cite{Shao:1997:ATLMS,RrWy1989}  For a subset of $\ms$ channels, the reduced
rank form of this criterion is
\[\BIC{} = \ln\MSEP + \frac{\log\ro}{\ro - \fc - 1}\ms,\]
where \MSEP is based on the given set of $\ms$ channels and $\fc$ factor PLS.
I use the denominator term $\ro-\fc-1$, rather than $\ro - \ms$ as
suggested by Shao for the OLS formulation, based on a simulation study.  This
allows meaningful comparison in situations where $\ms > \ro,$ although in this
case the expected value of \MSEP is penalized by a term quadratic in 
$\ms / \ro$.  \cite{NbCrr2005}
To continue the mongrelization of this criterion, I find it useful to replace
\MSEP by \MSECV for appropriately chosen CV groups:
\[\aBIC{} = \ln\MSECV + \frac{\log\ro}{\ro-\fc-1}\ms.\]
Minimization of this criterion favors parsimony more than minimization of
\RMSECV alone.  Until the asymptotic consistency of the reduced rank/PLS
model selection problem is addressed theoretically, I cannot recommend one of
these criteria over the other.

It is not obvious that the magnitudes of the elements of \vlam are sufficient
to establish an importance ordering on the channels.  For instance, it might be
appropriate to multiply the elements of \vlam by the corresponding element of
the regression vector $\regb$ chosen by WPLS on the entire data set, and use that
Kronecker product vector as the importance ordering.  It might be argued that
that product should further be multiplied by the sample standard deviation of
the channels.  As there seems to be no
general trend in comparing the two methods, I recommend implementing each of
these techniques and allowing the information criterion to select whichever model is
best, irrespective of which pragma produced the ordering.


\section{Crafting an Objective Function}\label{sec:objfun}

The ultimate goal is selection of a subset of the channels which minimizes 
delete-$d$ \RMSECV{} or one of the information criteria.  This should guide the
choice of the objective function which we numerically minimize in the continuous 
framework.  The obvious choice is to minimize \RMSECV{}, however the choice of
the CV groups can lead to an asymptotically inconsistent selection
procedure or long runtime.  
Moreover, the minimization 
of \RMSECV{} may also select a \vlam with a large number of nontrivial elements, which makes reversing the embedding difficult or noninformative.  

Thus one may choose to minimize an
objective function which approximates one of the information criteria,
balancing quality and parsimony, rather than minimizing
\RMSECV{}. 
Recall, for example, 
\(\aBIC{} = \ln\MSECV + \wrapparens{\ms\,{\log\ro}}/\wrapparens{\ro - \fc - 1}\).
The continuous embedding of the \MSECV term with respect to \vlam is understood.
To complete the embedding it only remains to estimate the subset size (number
of channels retained) of the model indicated by a continuous predictor
weighting \vlam.
%
%
%
%
My suggestion is to use a ratio of the $p$-norm to the $q$-norm:
\[\mrpq[p,q]{\vlam} =
\Parens{\frac{\norm[p]{\vlam}}{\norm[q]{\vlam}}}^{pq/(q-p)},\qquad\text{where }\,
\norm[p]{\vlam} = \Parens{\sum_j \abs{\lame{j}}^p}^{1/p},\,
\] 
for $0 < p < q < \infty.$ I claim \mrpq{\vlam} is an appropriate choice of 
model size estimate.  Note that \mrpq{} is scale invariant, that is,
\mrpq{\sclr \vlam} is constant for each nonzero choice of the scalar $\sclr.$
Also note that if $\vlam$ consists of $j$ ones and $n-j$ zeros, then
$\mrpq{\vlam} = j$.  See \figref{clover} for the behaviour of this function for
various values of $p, q.$  Using smaller values of $p$ creates a stronger
drive towards binary vectors by penalizing non binary vectors.


\begin{figure}[htb!]
\centering
	\psfrag{255}[][][0.85]{$\mrp[0.25,0.75]{}$}
	\psfrag{11}[lt][rt][0.85]{$\mrp[1,1.01]{}$}
	\psfrag{125}[][][0.85]{$\mrp[1,2.5]{}$}
	\psfrag{52}[][][0.85]{$\mrp[2,5]{}$}
	\psfrag{29}[rb][rb][0.85]{$\mrp[19,20]{}$}
	\includegraphics[angle=0,height=.62\columnwidth,width=.65\columnwidth]{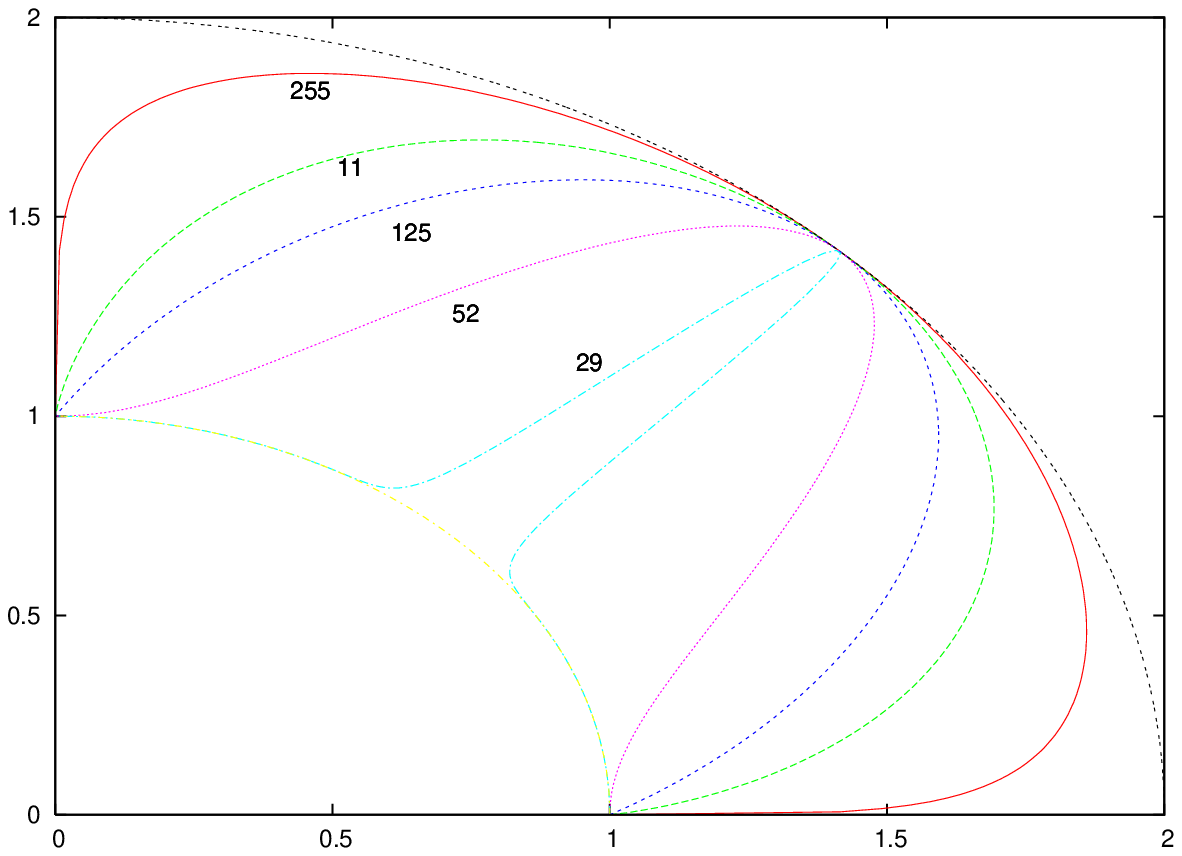}
\caption{The function \mrpq{\vlam} is plotted in polar coordinates for the two
dimensional vector $\vlam = \Angles{\cos{\theta},\sin{\theta}},$ and various
values of $p, q$ in the quadrant $0 \le \theta \le \pi/2$.  Note that each of 
the graphs passes through \tuple{0,1}, \tuple{\pi/4,2} and \tuple{\pi/2,1}, as
guaranteed by the fact that \mrpq{} measures the number of nonzero elements in
a scaled binary vector.  The circles of radius $1$ and $2$ are also plotted.}\label{fig:clover} 
\end{figure}

%

Thus to approximately minimize \BIC, one could minimize
\[\obf{\vlam} \defeq \ln\Parens{\MSECV\Parens{\vlam}} + (\ln\Parens{\ro}\,
\mrpq{\vlam})/(\ro-\fc-1),\]
the gradient of which is
\[
\glam{\obf{\vlam}} =
\frac{\glam{\MSECV\Parens{\vlam}}}{\MSECV\Parens{\vlam}} +
(\ln\Parens{\ro} \glam{\mrpq{\vlam}}) / (\ro-\fc-1).
\]


%
%

\section{Implementation Notes}\label{sec:impnotes}

The method was implemented in the Matlab\textsuperscript{\textsc{tm}}{}
language.  All results in this paper were performed in the 
GPL Matlab clone, GNU Octave
(version 2.1.69) \cite{eaton2001octave}, compiled with BLAS \cite{Dj2002}, on an AMD 
Athlon 64 4000+ (2.4 G\Hz) running Gentoo Linux, 2.6.15 kernel.

The BFGS and backtracking line search algorithms were implemented as outlined
by Nocedal and Wright.  \cite{NjWsj1999}  Sample code for the PLS
and Jacobian computation are given in \secref{code}.  The objective function
was supplemented by an optional term to simulate \BIC{}, with the $p$ and $q$
terms of \mrpq{} tunable parameters.  
The inverse of the sample standard deviation of each channel is generally used
as the starting iterate, \vlamk[0].  The initial approximation to the Hessian
is taken as some constant times the identity matrix.  Termination of BFGS was
triggered by the computation of a gradient smaller in norm than a lower bound,
by an upper limit on the number of major iterates, function evaluations, or
achievement of a lower bound on the change of the objective function.
Response weighting (\ie WPLS) has not yet been implemented.

Selection of optimal wavelengths was performed by minimization of a delete-$d$
\RMSECV{} or \aBIC{}, with the different models determined by ordering of \vlam
or by $\diag{\regb}\vlam$, whichever is chosen by the information criterion. 
The trivial model (responses are normal with approximated means and variances
and predictors are ignored) is also compared. 
An optional post-selection minimization is allowed on the selected channels.
The final results consist of a subset of the available channels and predictor
weightings for those channels.   This bit of cheating allows the method to keep
the advantages of properly weighted predictors.  Note that the weighting is
irrelevant and ignored in the case where the number of selected channels is
equal to the number of latent factors.

%
Unless specified by the user, the system must choose the cross validation
groups.  By default, this is done using the Monte Carlo framework suggested by Shao:
$2\ro$ different partitions of $\ro^{3/4}$ calibration
objects and $\ro - \ro^{3/4}$ test objects are randomly selected.

\section{Experiments and Results}

\SRCEK{} was tested on the three data sets used to compare wavelength
selection procedures by Forina \etal  \cite{Forinaetal2004,ForinaPC2006}
As in the original publication, these are referred to as 
\dsm, \dsk, and \dsa.  The data set \dsm{}
consists of moisture responses of 60 samples of soy flour, with predictors
measured with a filter instrument and originally appeared in a paper by Forina
\etal \cite{Forinaetal1995}.  The data set \dsk{}, originally from a paper by
Kalivas \cite{Kjh1997}, consists of moisture responses of 100 samples of wheat
flour, with 701 responses measured by an NIR spectrometer.  The data set \dsa{}
consists of 400 randomly generated objects with 300 channels.  The channels
are grouped into six classes, with a high degree of correlation between
elements of the first five classes, each consisting of 10 channels.  The response was 
generated by a linear combination of five of the responses (the first response
in each of the first five classes), plus some noise; the 250 channels of
the sixth class are entirely irrelevant to the responses.  However, the level
of noise in the response is large enough to mask the effects of the fifth
relevant channel.  The objects are divided into a training set of 100
objects, and an external set with the remainder.

In order to allow meaningful comparison between the results found here and in
the study of Forina \etal, I report \RMSECV values using the same CV
groupings of that study.  These were generated by dividing the objects into
groups in their given order.  Thus \eg the first group consists of the
\skth{1}, \skth{6}, \skth{11} objects and so on, the second group is the
\skth{2}, \skth{7}, \skth{12} objects and so on, \etc  Five groups are used.
\cite{Forinaetal2004,ForinaPC2006}  However, the objective function was computed
based on other groupings, as described below.

Note that, in light of Shao's studies, the CV groupings used by Forina 
\etal seem sparse both in number and in the number deleted ($\ro/5$). 
For this reason, it may be meaningless to compare the different subset selection 
techniques based on the \RMSECV{} for this grouping.  However, since
the channels retained by the different methods are not reported for the data
sets \dsk{} and \dsa, I can only compare the results of \SRCEK to
those of the other methods by the \RMSECV{} of this grouping or by the \aBIC
based on that \RMSECV{}.  For \dsm,
Forina \etal report the selected channels, making comparison based on
Monte Carlo CV groupings possible.  These are denoted by \RMSEMCCV, and based on
120 delete-$38$ groupings.

\SRCEK{} was applied to \dsm{} with 120 delete-$38$ Monte Carlo CV groups,
using \RMSEMCCV as
the objective function for 2 and 3 factor PLS.  Termination was triggered by
small relative change in the objective (relative tolerance in objective of
\tenex{}{-5}), which was achieved in both cases in at most 25 major iterates.  
Model selection is performed by minimization of \aBIC.
Results are summarized in \tabref{dsmtable},
and compared to the results found by Forina \etal  For 2 factor PLS, \SRCEK{}
selects 2 channels, L14:2100 and L16:1940, the same choice made by SOLS and
GA-OLS.  For 3 factor PLS, \SRCEK{} selects 3 channels, adding L20:1680 to
the previous two channels.  The 2 factor choice is preferred by both the CV
error and \aBIC.
In this case, \aBIC matches my intuitive ordering of these results.

\begin{table}[htb!]
\begin{center}
\footnotesize

\begin{tabular}{rrrll|ll}
			method & \fc & k & \tiny{\RMSECV} & \aBIC & \tiny{\RMSEMCCV} & \aBIC\\
\hline
             PLS & 2 & 19 & 1.369 & 1.99 & 1.438 & 2.09\\
             MUT & 2 & 8 & 1.354 & 1.18 & 1.407 & 1.26\\
       UVE norml & 2 & 7 & 1.350 & 1.10 & 1.403 & 1.18\\
          UVE 95 & 2 & 10 & 1.346 & 1.31 & 1.395 & 1.38\\
       GOLPE   I & 2 & 15 & 1.319 & 1.63 & 1.389 & 1.73\\
          UVE 90 & 2 & 12 & 1.338 & 1.44 & 1.387 & 1.52\\
             IPW & 3 & 3 & 1.308 & 0.756 & 1.361 & 0.836\\
       GOLPE III & 2 & 3 & 1.307 & 0.752 & 1.356 & 0.824\\
          MAXCOR & 2 & 10 & 1.268 & 1.19 & 1.318 & 1.27\\
             ISE & 2 & 2 & 1.264 & 0.613 & 1.311 & 0.686\\
       GOLPE  II & 2 & 6 & 1.256 & 0.887 & 1.298 & 0.953\\
        SOLS (5) & 2 & 2 & 1.203\footnotemark & 0.514 & 1.240 & 0.574\\
\hline
      \SRCEK & 2 & 2 & 1.203 & 0.514 & 1.240 & 0.574\\
      \SRCEK & 3 & 3 & 1.259 & 0.680 & 1.285 & 0.720\\
\end{tabular}
\end{center}
\caption{Results from selection methods applied to data set \dsm{} are shown ordered
by decreasing \RMSEMCCV (120 delete-$38$ MC groupings), with results from \SRCEK{}.  Adapted from the study of Forina \etal \cite{Forinaetal2004} The number of retained channels
is indicated by $k$, the number of latent factors is \fc.  The \RMSECV based on
the groupings of Forina \etal are also given. Two \aBIC{} values are reported,
based on the Forina and Monte Carlo CV groupings. 
\footnotesize{\thefootnote{}. Last digit apparently misreported 
by Forina \etal}}\label{tab:dsmtable}
\end{table}

The results of \SRCEK{} applied to \dsk{} are summarized in \tabref{dsktable},
and compared to the results from the previous study.  Several experiments were
attempted, each using 3--5 factors.  The first experiment, (a), uses the same
CV groupings as Forina \etal, and minimizes and selects based on \RMSECV for
this grouping.  In experiments (b) and (c), 240 delete-$68$ MC CV groups are
used, \RMSECV is minimized, and channel selection is based on, respectively, 
\RMSECV and \aBIC.  In experiments (d) and (e), 120 delete-$68$ MC CV groups
are used, the embedded \aBIC (using \mrpq[1,2]{}) is minimized, and selection
is based on, respectively, \aBIC and \RMSECV.  The final models of all
experiments used to compute \RMSECV for the same 200 delete-$68$ MC CV
grouping, to facilitate comparison.  The maximum acceptable number of
channels was taken to be $50$.

\begin{table}[htb!]
\begin{center}
\footnotesize
\begin{tabular}{rrrlr|l}
method&\fc&k&\tiny{\RMSECV}&\aBIC&\tiny{\RMSEMCCV}\\
\hline
%
PLS&5&701&0.2218&31.3&n/a\\   
MAXCOR&5&684&0.2217&30.5&n/a\\   
UVE 95\%&5&657&0.2227&29.2&n/a\\   
GOLPE I&6&648&0.2216&29.1&n/a\\   
MUT&6&575&0.2227&25.5&n/a\\   
GOLPE II&6&352&0.2167&14.4&n/a\\   
GOLPE III&6&32&0.2231&-1.42&n/a\\   
LASSO&14&14&0.2153&-2.31&n/a\\   
VS&6&14&0.2111&-2.42&n/a\\   
IPW 3&3&11&0.2174&-2.52&n/a\\   
IPW 2&2&11&0.2155&-2.55&n/a\\   
GAPLS&6&11&0.2078&-2.60&n/a\\   
ISE&2&7&0.2151&-2.74&n/a\\   
SOLS(2)&2&2&0.2408&-2.75&n/a\\   
SOLS(4)&4&4&0.2207&-2.83&n/a\\   
GAOLS a&4&4&0.2154&-2.88&n/a\\   
GAOLS b&4&4&0.2090&-2.94&n/a\\   
\hline 					                         
\SRCEK (a) &3&49&0.2007&-0.861&0.2166\\ 
\SRCEK (a) &4&50&0.1869&-0.931&0.3444\\ 
\SRCEK (a) &5&40&0.1843& -1.42&0.2173\\ 
\SRCEK (b) &3&50&0.2013&-0.807&0.2153\\ 
\SRCEK (b) &4&49&0.1914&-0.931&0.2090\\ 
\SRCEK (b) &5&49&0.1848&-0.976&0.2074\\ 
\SRCEK (c)  &3&9&0.2118& -2.67&0.2254\\ 
\SRCEK (c)  &4&8&0.1992& -2.84&0.2218\\ 
\SRCEK (c)  &5&5&0.2093& -2.88&0.2350\\ 
\SRCEK (d)  &3&3&0.2479& -2.65&0.2643\\ 
\SRCEK (d)  &4&4&0.2318& -2.73&0.2500\\ 
\SRCEK (d) &5&31&0.2626& -1.16&0.2936\\ 
\SRCEK (e) &3&10&0.2187& -2.56&0.2344\\ 
\SRCEK (e) &4&18&0.2079& -2.27&0.2301\\ 
\SRCEK (e) &5&31&0.2626& -1.16&0.2936\\ 
\end{tabular}
\end{center}
\caption{Results from selection methods applied to data set \dsk{} are shown ordered
by decreasing \aBIC, with results from \SRCEK{}.  
The results from \SRCEK{} are also tested against a MC CV grouping
consisting of 200 delete-$68$ partitions, yielding the \RMSEMCCV{} shown.}\label{tab:dsktable}
\end{table}

As expected, when trained on the CV groups of Forina \etal, \SRCEK is able to
produce small errors for that CV grouping, beating all the methods studied by
Forina \etal.  A number of caveats are necessary:  the \RMSECV{} values
reported use predictor weighting to build and test the models.  When the
weights are not used, the \RMSECV{} values are not as impressive.  For example,
for experiment (a), $4$ factors, the reported $0.1869$ becomes $0.2171$ when
the predictor weighting is not used.   I think the objection here should not be 
that \SRCEK uses predictor weighting, but that the methods studied previously
did not, which puts them at a disadvantage when compared to \SRCEK.
A more serious objection is that a small
\RMSECV{} for the CV groupings of Forina \etal does not appear to imply a small
\RMSECV{} for the MC CV groupings, although the inverse implication does seem
to hold.  This gives confidence in the results of \eg experiment (b)-$5$, which
gives small \RMSECV{} for both CV groupings.  


The effect of the objective function on the algorithm outcome for this
data set is shown in \figvref{blams}.  This graph shows the effective regression
vector $\LAM\regb\Parens{\vlam},$ for the \vlam found by BFGS minimization, for
experiments (b) and (d), using $4$ latent factors.  
When \RMSECV alone is minimized, the regression vector has no clear
structure.  However, when \aBIC is minimized, the regression vector divides the
channels into two groups, those with `low' relevance, and those with `high'
relevance.  As expected from spectroscopic data, the relevance of relevant
channels is more or less continuous.  Note, however, that minimizing on the
information criterion selects some of the same channels as minimizing on
\RMSECV{}, but this relationship does not strictly hold.  For example, some
channels in the range 1-10 appear to be given high relevance by \aBIC but not
by \RMSECV{}.  I suspect that there is some dependence on the initial vector
\vlam, and the CV groups used.

\begin{figure}[htb!]
\centering
		\psfrag{blam1}[rb][rb][0.5]{$\LAM\regb\Parens{\vlam}$}
		\psfrag{blam2}[rb][rb][0.5]{$\LAM\regb\Parens{\vlam}$}
		\includegraphics[width=.45\columnwidth,angle=270]{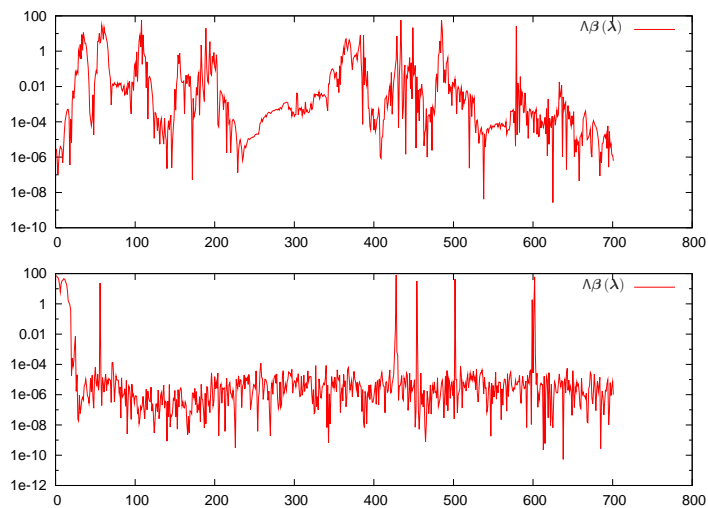}
\caption{The absolute value of the elements of vector $\LAM\regb\Parens{\vlam}$
are plotted for post-optimization \vlam on data set \dsk.  At top, \RMSECV for
MC CV groups was used as the objective function (experiment (b), $4$ latent
factors).  At bottom, the embedded \aBIC using the same CV groups for the
\RMSECV part and \mrpq[1,2]{}, was minimized (experiment (d), $4$ latent
factors).}\label{fig:blams}
\end{figure}



The results of \SRCEK{}, and other methods, applied to \dsa{} are summarized in 
\tabref{dsatable}. Two experiments were
attempted, each using 3--5 factors.  Each use autoscaling to generate the 
initial \vlam, and 120 delete-$75$ MC CV groups for computing \RMSECV.  
In both experiments, 
the embedded \aBIC (using \mrpq[0.8,2.4]{}) is minimized, and selection is based
on \aBIC and \RMSECV, respectively, in experiments (f) and (g).
The models built by these experiments are tested on a Monte Carlo CV grouping 
(200 delete-$68$ groups consisting of the 100 test data, not the external set),
and the computed \RMSEMCCV is shown in the table.   Note that there is some
mismatch between the \RMSECV for the groups used by Forina \etal and the
\RMSEMCCV for these groupings.  These measures produce different orderings of
the models, casting some suspicion on the sparse CV groupings used by the
previous study.

The post-optimized effective regression vector, $\LAM\regb\Parens{\vlam},$ is plotted in
\figvref{goodart} for this experiment (g) with $3$ factors, showing the discovered 
relevance of the channels.  All but one of the 250 irrelevant channels for this 
data set are found to have a low weight.  Selection by \RMSECV{} picks $6$
channels, numbers $4, 5, 10, 20, 24, 40$, which includes channels from four
of the five relevant correlated groups of channels.  The $5$ factor
experiment picks channels $1, 2, 3, 4, 5, 9, 10, 11, 13, 14, 15, 19, 20, 24,
40, 50$, which includes channels from each of the five correlated relevant
groups, and no irrelevant channels.  This may be attributable only to chance,
however, as the optimal effective regression vector $\LAM\regb\Parens{\vlam}$
after optimization in this experiment attributes high weights to a number of
irrelevant channels.

\begin{table}[htb!]
\begin{center}
\footnotesize
\begin{tabular}{rrrlc|rl|l}
method&\fc&k&\tiny{\RMSECV}&\aBIC&$\text{k}_{\text{bad}}$&ext
\tiny{\RMSEP}&\tiny{\RMSEMCCV}\\
\hline
GOLPE I &3&279&207.6&24.1&229&229.1&n/a\\   
GOLPE II&3&142&141.1&16.7&102&207.9&n/a\\   
ISE&3&109&122.2&14.8&63&206.7&n/a\\   
UVE&3&67&139.5&13.1&17&189.6&n/a\\   
MUT&3&59&147.8&12.8&9&184.5&n/a\\   
ISE&3&61&121.3&12.5&17&189.1&n/a\\   
MAXCOR&3&42&171.7&12.3&2&192.6&n/a\\   
GOLPE III&4&34&166.1&11.9&12&199.8&n/a\\   
SOLS(20)&20&20&106.1&10.5&13&240.5&n/a\\   
GAPLS&4&17&122.2&10.4&11&214.0&n/a\\   
SOLS(4)&4&4&166.9&10.4&0&195.3&n/a\\ 
IPW&4&4&163.3&10.4&0&192.6&n/a\\   
SOLS(5)&5&5&154.8&10.3&0&180.5&n/a\\   
GAOLS&10&10&126.0&10.2&3&218.3&n/a\\   
\hline
\SRCEK (f)&3&4&163.8&10.4&0&194.8&172.2\\ 
\SRCEK (f)&4&4&163.3&10.4&0&192.7&170.9\\
\SRCEK (f)&5&5&164.6&10.5&0&192.6&175.2\\
\SRCEK (g)&3&6&162.7&10.5&0&194.1&171.5\\ 
\SRCEK (g)&4&29&146.3&11.4&14&183.8&161.2\\
\SRCEK (g)&5&16&152.9&10.8&0&179.9&164.6\\
\end{tabular}
\end{center}
\caption{Results from selection methods applied to data set \dsa{} are shown ordered
by decreasing \aBIC, with results from \SRCEK{}.  The number of uninformative
channels selected is shown as well as the \RMSEP{} for the external set of 300
objects. The results from \SRCEK{} are also tested against a MC CV grouping
consisting of 200 delete-$68$ partitions, yielding the \RMSEMCCV{} shown.  The
objective function used was the embedded \aBIC.}\label{tab:dsatable}
\end{table}

\begin{figure}[htb!]
\centering
		\psfrag{blam}[rb][rb][0.5]{$\LAM\regb\Parens{\vlam}$}
		\includegraphics[width=.45\columnwidth,angle=270]{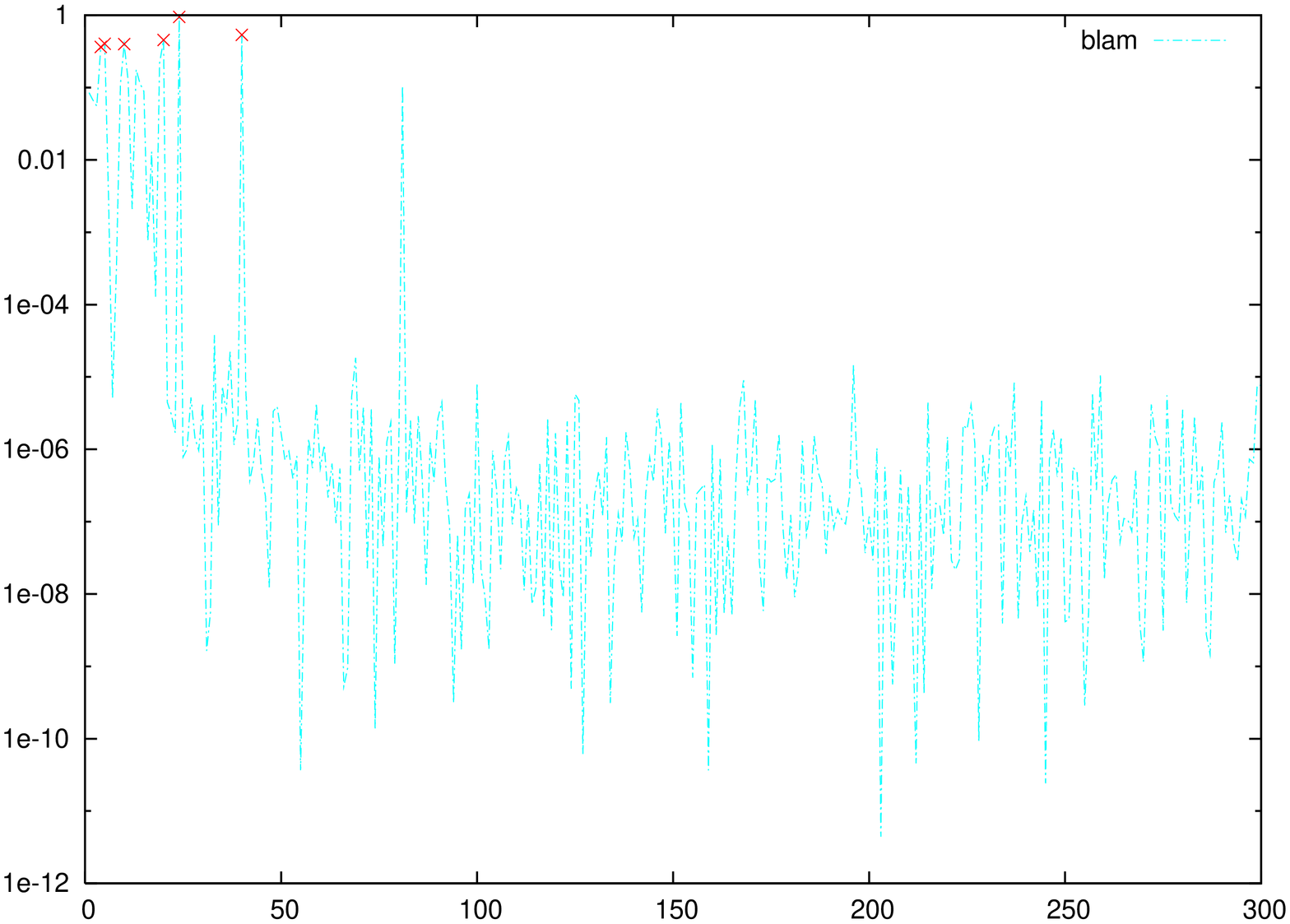}
\caption{The vector $\LAM\regb\Parens{\vlam}$ is plotted for post-optimization
\vlam on data set \dsa{} from \SRCEK{} (3), experiment (g), using embedded \aBIC
as objective function.  The 6 selected (by \RMSECV{}) channels are indicated with 
crosses.  Many of the 50 relevant channels are highly weighted, while most of
the 250 irrelevant channels have very low weights, with one notable
exception.  That the highly weighted irrelevant channel was not selected may
be attributable only to chance.}\label{fig:goodart}
\end{figure}



\section{Directions for Further Study}

Foremost, it seems one should be able to optimize \RMSECV{} with respect to
response weightings, \vgam, in addition to predictor weights \vlam.  
One can easily alter \algref{fullonpls} to also compute the gradients with
respect to \vgam.  The increased degrees of freedom increases the risk of
overfitting.   One should alter the embedded information criterion objective function
described in \secref{objfun} to balance this risk.
Since it is assumed the data are distributed as 
\(\ysym_j \sim \trans{\vect{\Xsym}_j}\regb + \rego + \mathcal{N}\Parens{0,\sigma^2/\game{j}},\)
we have added $\ro - 1$ new estimated parameters, \viz the separate variances
of each observation.  \cite{Kg1979,Ps1992} One strategy to embed the
information criterion, then, is to let
\(\obf{\vlam} = \ln\Parens{\MSECV\Parens{\vlam}} + (\mrpq{\vlam} -
\mrpq{\vgam}) \ln\Parens{\ro} / (\ro - \fc - 1).\)  The initial estimate 
should be that of 
homoscedasticity.  Comparison of models becomes tricky.  Work is underway on
this extension.

A theoretical study of the asymptotic consistency of different model selection
techniques for the case of reduced rank design matrix and PLS modeling would
provide \SRCEK{} with a more sound method for reversing the embedding, as well
as a better objective function.  As this problem seems intractable, a
simulation study might be appropriate.  If this theoretical study reveals that
one should minimize some combination of \RMSECV (for some tailored CV
groupings) with model size, I am confident that that measure can be
continuously embedded using the techniques of this note.


The method of ordering the successively larger models based on the optimal
\vlam, or on the Kronecker product of \vlam with the regression coefficient
seems rather \adhoc{}.  This step would also benefit from some theory, or could
perhaps be replaced by strategies
cribbed from other channel selection techniques (\eg IPW).
Conversely, some of these techniques may benefit from a preliminary predictor
weight optimization step via BFGS.

The \SRCEK{} method could also be extended to Kernel PLS regression, however I
suspect this would require the kernel to also compute derivatives, which
could be impractical.  \cite{HLetal2003}



I would be interested in an analysis of the structure of the
\RMSECV{} and embedded \aBIC objective functions.  For example, can either be shown 
to be convex (in \vlam) in general, or under assumptions on the data \Xk[]
and \yk[] which are justifiable in the chemometric context?  Moreover, can one
find sufficient sizes for the CV groupings to sufficiently reduce dependence of the
objective on the groupings? Will a sufficiently designed CV grouping make the
objective function convex or nearly so?

The choice of CV groupings affects both the ``quality'' of the RMSECV measure
(\ie how accurately it rates subsets of channels) and the runtime of \SRCEK{}.
Shao's Monte Carlo scheme, using $2\ro$ groupings of $\ro^p$ objects in the
calibration group, and the remainder in the test group, results in a total
runtime of \bigo{\ro\Parens{\ro^{2p}\cl(\fc-1) +
\ro^{p+1}\cl}} flops for the computation of the gradient of \RMSECV.  Thus I would
be interested to discover an acceptable lower bound on $p$ which gives
acceptable quality of \RMSECV.

The iterates of the BFGS algorithm for this objective function often display a
zigzagging behaviour towards the minimum.  Often this is the result of some
elements of \vlam ``overshooting'' zero.  It would be interesting to see if
this can be avoided by using other minimization techniques, for example the 
conjugate gradient method, or a proper constrained minimization implementation
of BFGS.  \cite{nocedal91theory}

Finally, the \SRCEK{} method as described in this paper has many tweakable
parameters: initial iterate \vlamk[0], initial approximation of the Hessian,
termination condition for BFGS, choice of objective function, $p$ and $q$ for
the continuous embedding of number of estimated parameters, \etc  While these
provide many possibilities to the researcher of the technique, they are an
obstacle for the end user.  Thus reasonable heuristics for setting these
parameters which work well in a wide range of settings would be welcome.
\bibliographystyle{plainnat}
\bibliography{wsel}
\pagebreak
\appendix

\section{A Brief Review of Vector Calculus}\label{sec:bguide}

The vector calculus appearing in this paper is not difficult, it is merely
involved.  To follow the derivations one need only have a good
understanding of basic linear algebra and a fondness for the product derivative rule.

First the gradient and Jacobian:  given a scalar function $f(\vlam)$, its
gradient is the vector $\glam[f]$ whose \kth{j} element is the partial derivative of 
$f$ with respect to \lame{j}, that is \prbypr{f}{\lame{j}}.
Given a vector valued function $\vect{f}(\vlam)$ which outputs an
$l$-dimensional vector when given the $m$-dimensional vector \vlam as input,
its Jacobian is an $l\cross m$ matrix whose \kth{(i,j)} element is the partial
derivative of the \kth{i} component of \vect{f} with respect to \lame{j}, that
is
\[\Parens{\dlam{\vect{f}}}_{i,j} = \prbypr{f_i}{\lame{j}}.\]

Remarkably we need nothing more exotic than this.  Some convenient rules to
fill our toolbox:
\begin{compactenum}
\item The gradient product rule: given functions $f$ and $g$ we have
\begin{equation}
\glam\Parens{fg} = f\glam{g} + g\glam{f}.
\label{eqn:gprodrule}
\end{equation}
This follows because the \kth{j} element of \glam[\Parens{fg}] is the partial of
$fg$ with respect to \lame{j}.  Using the scalar product rule we have
$\prbypr{fg}{\lame{j}} = f\prbypr{g}{\lame{j}} + g \prbypr{f}{\lame{j}},$
from which our rule follows.

The gradient quotient rule follows in a similar fashion: for given $f$ and $g$:
\begin{equation}
\glam\Parens{f/g} = \frac{g\glam{f} - f\glam{g}}{g^2}.
\label{eqn:gquotrule}
\end{equation}

The gradient dot-product rule is similar; given vectors \vect{v} and \vect{w}
we have
\begin{equation}
\glam\Parens{\ip{\vect{v}}{\vect{w}}} = \trans{\dlam{\vect{v}}}\vect{w} +
\trans{\dlam{\vect{w}}}\vect{v}.
\label{eqn:gdprodrule}
\end{equation}

\item The Jacobian scales: given a constant matrix \Mtx{M} and a vector-valued function
\vect{f} we have
\begin{equation}
\dlam{\Mtx{M}\vect{f}} = \Mtx{M}\dlam{\vect{f}}.
\label{eqn:jprodrule}
\end{equation}
This follows from the linearity of the matrix product and the derivative:
\[\Parens{\dlam{\Mtx{M}\vect{f}}}_{i,j} 
= \prbypr{\Parens{\Mtx{M}\vect{f}}_{i}}{\lame{j}}
= \prbypr{\sum_k M_{i,k} f_k}{\lame{j}}
= \sum_k M_{i,k} \prbypr{f_k}{\lame{j}}
= \sum_k M_{i,k} \Parens{\dlam{\vect{f}}}_{k,j}
= \Parens{M \dlam{\vect{f}}}_{i,j}
\]
\item The diagonal rule: given vector-valued function \vect{f}, and letting
$\LAMk{p} = \Parens{\diag{\vlam}}^p$ then
\begin{equation}
\dlam{\LAMk{p}\vect{f}} = p\LAMk{p-1}\diag{\vect{f}} + \LAMk{p} \dlam{\vect{f}}.
\label{eqn:diagrule}
\end{equation}
Again the proof is trivial:
\begin{align*}
\Parens{\dlam{\LAMk{p}\vect{f}}}_{i,j} 
&= \prbypr{\Parens{\LAMk{p}\vect{f}}_{i}}{\lame{j}}
= \prbypr{\sum_k \Parens{\LAMk{p}}_{i,k} f_k}{\lame{j}}
= \prbypr{\lamex{p}{i} f_i}{\lame{j}}\\
&= \lamex{p}{i} \prbypr{f_i}{\lame{j}} + f_i \prbypr{\lamex{p}{i}}{\lame{j}}
= \lamex{p}{i} \Parens{\dlam{\vect{f}}}_{i,j} + f_i p\lamex{p-1}{i}\kron{i,j}\\
&= \Parens{\LAMk{p} \dlam{\vect{f}}}_{i,j} + \Parens{p\LAMk{p-1}\diag{\vect{f}}}_{i,j}.
\end{align*}
We have used $\kron{i,j}$ to be the Kronecker delta, which is one if $i=j$ and
zero otherwise.

\item The product rule for Jacobians: given scalar function $f$ and vector
function \vect{g}, we have
\begin{equation}
\dlam{\Parens{f\vect{g}}} = f \dlam{\vect{g}} + \vect{g}\trans{\Parens{\glam[f]}}.
\end{equation}
The proof is similar to above.

\item A useful composite rule; given vectors \vect{v} and \vect{w} and matrix
\Mtx{M}, then
\begin{equation}
\dlam{\Parens{\trans{\vect{v}}\Mtx{M}\LAMk{p}\trans{\Mtx{M}}\vect{w}}} = 
\trans{\dlam{\vect{v}}}\Mtx{M}\LAMk{p}\trans{\Mtx{M}}\vect{w} + 
p\, \diag{\trans{\Mtx{M}}\vect{w}}\LAMk{p-1}\trans{\Mtx{M}}\vect{v} +
\trans{\dlam{\vect{w}}} \Mtx{M}\LAMk{p}\trans{\Mtx{M}}\vect{v}.
\end{equation}
This rule follows from application of the dot-product and diagonal rules given
above.

\end{compactenum}

\section{Code}\label{sec:code}
I present Matlab\textsuperscript{\textsc{tm}}{} compatible code for computing
the preimage of the PLS regression coefficient and its Jacobian on the
following page.  This code is merely a realization of \algref{fullonpls}.
The code was tested in GNU Octave, and used to produce some of the results depicted
in \figvref{comprtimes}.
\pagebreak

\begin{verbatim}
function [alpha,dalpha,b0,db0]	= fastpreplsandjacobian(X,y,lambda,gam,l)
% [alpha,dalpha,b0,db0]	= fastpreplsandjacobian(X,y,lambda,gam,l)
%
% code to compute the preimage of the l-factor PLS regression coefficient 
% to fit the model y \approx X diag(lambda) diag(lambda) X' alpha
% also computes the jacobian of the preimage alpha with respect to lambda.
% and the intercept and its gradient wrt lambda.
%
% input:
%  X        an m by n matrix.             y    an m by 1 vector.
%  lambda   n vec prdctr weights.       gam    an m vec objct weights.
%  l        the number of pls factors.
% output:
%  alpha    a m by 1 vector of the regression coefficients.
%  dalpha   a m by n matrix of the Jacobian of beta wrt lambda.
%  b0       the scalar intercept (centers data)
%  db0      the gradient of same.
% nb: assumes m << n
% Author: Steven Pav % Created: 2006.04.11 % Copyright 2006
% sanity checking:
[m,n] = size(X);lambda = lambda(:);gam = gam(:);y = y(:);
if ((rows(y) != m) || (length(lambda) != n) || (length(gam) !=m)) 
  error('size mismatch.'); end
%allocate storage
V  = zeros(m,l);dV = zeros(m,n,l-1);w = zeros(l,1);dw = zeros(l,n);
q  = zeros(l+1,1);dq = zeros(l+1,n);
%used so much we compute once and store.
Xlam = X * diags(lambda);XLXg = Xlam * Xlam' * diag(gam); %nm^2 hit.
T0 = Tk = ones(m,1);dTk = zeros(m,n);Vk = y;dVk = zeros(m,n);
t0 = sum(gam);
for k=0:l
  rk = y' * (gam .* Tk);  drk = dTk' * (gam .* y);
  tk = Tk' * (gam .* Tk); dtk = 2 * dTk' * (gam .* Tk);
  q(k+1) = rk/tk; dq(k+1,:) = (drk - q(k+1) * dtk)' / tk;
  if (k < l)
    Vk = Vk - q(k+1) * Tk;V(:,k+1) = Vk;
    dV(:,:,k+1) = (dVk = dVk - q(k+1)*dTk - Tk * dq(k+1,:));
    Qk = XLXg * Vk;dQk = XLXg * dVk + 2 * Xlam * diag(X' * (gam .* Vk));
    Uk = Qk .- (gam' * Qk / t0);dUk = dQk - T0 * (gam' * dQk / t0);
    w(k+1) = Tk' * (gam .* Uk) / tk;
    dw(k+1,:)  = (dTk' * (gam .* Uk) + dUk' * (gam .* Tk) - w(k+1) .* dtk)' / tk;
    dTk = dUk - w(k+1) * dTk - Tk * dw(k+1,:);
    Tk = Uk - w(k+1) * Tk;
  end
end

%now compute M\q and its jacobian
iMq = zeros(l,1);diMq = zeros(l,n);
iMq(l) = q(l+1);diMq(l,:) = dq(l+1,:);

for k=(l-1):-1:1
  iMq(k)    = q(k+1) - w(k+1) * iMq(k+1);
  diMq(k,:) = dq(k+1,:) - (w(k+1) .* diMq(k+1,:) + iMq(k+1) .* dw(k+1,:));
end

%now compute VM\q and its jacobian
alpha = V * iMq;dalpha = V * diMq;      %+ more stuff:
%ack! no tensor product in octave/Matlab :(
for k=2:l    dalpha  += iMq(k) .* dV(:,:,k); end %dV(:,:,1) is all zeros?

%now the intercept
b0  = gam' * (y - XLXg * alpha) ./ sum(gam);
db0 = - (XLXg * dalpha + 2 * Xlam * diag(X' * (gam .* alpha)))' * gam ./ sum(gam);

endfunction
\end{verbatim}

\end{document}